\documentclass[final]{llncs} 
\pdfoutput=1

\usepackage{tabularx}
\usepackage{algorithmic}
\usepackage{algorithm}
\usepackage{amsmath}
\usepackage{amsfonts}
\usepackage{amssymb}
\usepackage[usenames,dvipsnames]{xcolor}
\usepackage{tkz-graph}
\usepackage{fullpage}
\usetikzlibrary{matrix,arrows,decorations.pathmorphing}

\title{Maximum Matching in Semi-streaming with Few Passes\thanks{Most of the work was done while Christian Konrad was a PhD student at LIAFA, Universit\'{e} Paris Diderot. Fr\'{e}d\'{e}ric Magniez is supported by the French ANR Blanc project ANR-12-BS02-005 (RDAM). 
Claire Mathieu is supported by NSF Grant CCF-0964037.}}

 \author{Christian Konrad \inst{1} \and Fr\'{e}d\'{e}ric Magniez \inst{2} \and Claire Mathieu \inst{3}} 

\institute{Reykjavik University, Reykjavik, Iceland \\ \email{christiank@ru.is} \and CNRS, LIAFA, Universit\'{e} Paris Diderot, Sorbonne Paris-Cit\'{e}, 75205 Paris, France \\ \texttt{magniez@cnrs.fr} \and CNRS, D\'{e}partement d'Informatique UMR CNRS 8548, \'{E}cole Normale Sup\'{e}rieure, Paris, France \\ \email{claire@cs.brown.edu}}

\newcommand{\mm}{\textsc{Maximum Matching}}
\newcommand{\mbm}{\textsc{Maximum Bipartite Matching}}
\newcommand{\mbmshort}{\textsc{MBM}}
\newcommand{\mmshort}{\textsc{MM}}
\newcommand{\OPT}{\mathrm{opt}}
\newcommand{\Greedy}{\mathrm{Greedy}}
\newcommand{\MM}{\OPT}
\DeclareMathOperator*{\Exp}{\mathbb{E}}
\newcommand{\Order}{O} 
\DeclareMathOperator*{\polylog}{polylog}

\DeclareMathOperator*{\argmin}{arg\,min}

\begin{document}

\maketitle

\begin{abstract}
In the {\em semi-streaming model}, 
an algorithm receives a stream of edges of a graph in arbitrary order and uses a memory of 
size $\Order(n \polylog n)$, where $n$ is the number of vertices of a graph. 
In this work, we present semi-streaming algorithms 
that perform one or two passes over the input stream for \mm{} with no restrictions on the input graph, 
and for the important special case of bipartite graphs that we refer to as \mbm{}. 
The Greedy matching algorithm performs one pass over the input and outputs a $1/2$ approximation. 
Whether there is a better one-pass algorithm has been an open question since the appearance of the 
first paper on streaming algorithms for matching problems in 2005 [Feigenbaum et al., SODA 2005]. 
We make the following progress on this problem:

In the one-pass setting, we show that there is a deterministic semi-streaming algorithm for \mbm{}
with expected approximation factor $1/2+0.005$, 
assuming that edges arrive one by one in (uniform) random order. We extend this algorithm to general 
graphs, and we obtain a $1/2+0.003$ approximation for \mm.

In the two-pass setting, we do not require the random arrival order assumption (the edge stream 
is in arbitrary order). We present a 
simple randomized two-pass semi-streaming algorithm for \mbm{} with expected approximation factor $1/2 + 0.019$.
Furthermore, we discuss a more involved deterministic two-pass semi-streaming algorithm for \mbm{} 
with approximation factor $1/2 + 0.019$ and a generalization of this algorithm to general graphs with
approximation factor $1/2 + 0.0071$.
\end{abstract}


\section{Introduction}
{\bf Streaming.} Classical algorithms assume random access to the input. This is a reasonable assumption until one 
is faced with massive data sets as for instance in bioinformatics for genome decoding, Web databases for the search of 
documents, or network monitoring. The input may then be too large to fit into the computer's memory. Another typical situation 
is a continuous flow of traffic logs sent to a router. Streaming algorithms sequentially scan the input piece by piece 
while using sublinear memory space. The analysis of Internet traffic \cite{ams99} 
was  one of the first applications of such algorithms. A similar but slightly different situation arises when the input is 
recorded on an external storage device like optical disks or hard drives where random access is too costly and hence only sequential 
access is possible. Then a small number of passes, ideally constant, can be performed.

By sublinear memory one ideally means memory that is polylogarithmic in the size of the input. However, polylogarithmic memory 
is too restrictive for many graph problems: as shown in \cite{fkmsz052}, deciding basic graph properties such as bipartiteness 
or connectivity of an $n$-vertex graph already requires $\Omega(n)$ space. Muthukrishnan~\cite{mut05} suggests to study massive
graphs in a {\em semi-external} model, that is, not the entire graph but the vertex set can be stored 
in memory. In that model, a graph is given by a stream of edges arriving in arbitrary order. A {\em semi-streaming} algorithm has 
memory $\Order(n \polylog n)$, and the graph vertices are usually known before processing the stream of edges.

{\bf Matchings.} In this paper, we focus on an iconic graph problem: finding large matchings. In the semi-streaming model, the problem 
was primarily addressed by Feigenbaum, Kannan, McGregor, Suri and Zhang~\cite{fkmsz05}. Currently, a variety of semi-streaming 
matching algorithms for particular settings exist (unweighted/weighted, bipartite/general graphs). Most works consider the 
multipass scenario \cite{ag11,eks09} where the goal is to find a $(1-\epsilon)$ approximation while minimizing the number of passes. 
The techniques are based on finding augmenting paths, and, recently, linear programming was also applied \cite{ag11}. Ahn and
Guha \cite{ag11} provide an overview of the current best algorithms. 

In this work, we focus on semi-streaming algorithms that perform one or two passes.
In the one-pass setting, in the unweighted case, the greedy matching algorithm is still the best known algorithm
(We note that there are algorithmic results in the weighted case  \cite{fkmsz05}, \cite{m05}, \cite{z08} \cite{elms10},
 but when the edges are unweighted those algorithms are of no help.).
The greedy matching algorithm constructs a matching in the following online fashion: starting  with an empty matching $M$, 
upon arrival of edge $e$, it adds $e$ to $M$ if $M \cup \{e \}$ remains a matching. A \textit{maximal matching} is a matching that 
 can not be enlarged by adding another edge to it. It is well-known that 
the cardinality of maximal matchings is at least half of the cardinality of maximum 
matchings. By construction, since the greedy matching is maximal, 
$M$ is a $(1/2)$-approximation of any maximum matching $M^*$, that is $|M|\geq |M^*|/2$.
The starting point of this paper is to address the following question: 

{\bf Is the greedy matching algorithm best possible, or is there a semi-streaming algorithm with an approximation ratio better than $1/2$? } 

This is an open question at least since the publication of the first paper \cite{fkmsz05} on computing matchings in the 
semi-streaming in 2005. 
On the negative side, Michael Kapralov showed in \cite{kap13} that there is no semi-streaming algorithm for maximum matching
(even for bipartite graphs) with approximation factor asymptotically better than $1 - 1/e$. 
This, however, still leaves room between $1/2$ and $1-1/e$.
To get an approximation ratio better than $1/2$, prior multipass semi-streaming algorithms 
require at least $3$ passes, for instance the algorithm of \cite{eks09} can be
used to run in $3$ passes providing a matching strictly better than a $(1/2)$-approximation.

{\bf Random order of edge arrivals.}
The behavior of the greedy matching algorithm has been extensively studied in a variety of settings. The most relevant 
reference~\cite{df91} considers a (uniform) random order of edge arrivals. In that setting, Dyer and Frieze showed that
the expected approximation ratio is still $1/2$ for some graphs, but can be better for particular graph classes such as 
planar graphs and forests.

In the context of streaming and semi-streaming algorithms, the model of random order arrival has first been studied for the 
problems of sorting and selecting in limited space by Munro and Paterson~\cite{mp80}. Guha and McGregor~\cite{gm09} gave an 
exponential separation between random order and adversarial order models. Recently, Kapralov, Khanna and Sudan showed in \cite{kks14}
that under the random order arrival assumption, the size of a maximum matching can be approximated within a constant
factor using only polylogarithmic space. One justification of the random order model is to 
understand why certain problems do not admit a memory efficient streaming algorithms in theory, while in practice, heuristics 
are often sufficient.

{\bf Other related work.}
Maximum bipartite matching was also intensively studied in the online setting, where nodes from one side arrive in adversarial 
order together with all their incident edges. In this model, the decision to take or discard an edge has to be taken before
accessing the edges of the next vertex. The well-known randomized algorithm by Karp Vazirani and Vazirani~\cite{kvv90} (KVV algorithm) achieves 
an approximation ratio of $1-1/e$ for bipartite graphs where all nodes from one side are known in advance, the nodes from the 
other side arrive online. They prove that their bound is optimal in the worst case. This barrier was broken only recently by modifying 
the worst case assumption (worst input graph and worst arrival order) to assume that, although the graph itself is worst-case, 
the arrival order is according to some (known or unknown) distribution~\cite{kmt11,my11}.

The online model for bipartite matching carries over to the streaming model. In the so-called {\em vertex arrival order} model,
the input edge stream is sorted with respect to the vertices of one bipartition \cite{gkk12,kap13}. The KVV algorithm can also be seen
as a streaming algorithm when the incoming edge sequence is in vertex arrival order. Surprisingly, it turns out that the KVV 
algorithm is {\em optimal}, that is, no semi-streaming algorithm for \mbm{} can achieve an approximation factor
better than $1-1/e$ \cite{kap13}. Goel, Kapralov and Khanna showed in \cite{gkk12} that there is a deterministic
counterpart to the KVV algorithm in the semi-streaming model that achieves the same approximation factor.
This separates the online setting from the vertex arrival order setting
in streaming since it is well-known that any deterministic online algorithm for \mbm{} cannot achieve an approximation ratio
better than $1/2$.


{\bf Our results.} 
In this paper, we present semi-streaming algorithms for maximum matching for bipartite graphs and general graphs 
with approximation factor strictly larger than $1/2$. 
Our algorithms make one or two passes over the input.
Our one-pass semi-streaming algorithm for bipartite graphs is deterministic and achieves an expected approximation 
ratio $1/2+0.005$ for any graph (\textbf{Theorem~\ref{theorem:one-pass-bip}}) assuming
that the edges arrive one by one in (uniform) random order. 
Furthermore, we modify the latter algorithm in order to obtain an algorithm for general graphs, 
and we obtain an approximation ratio of $1/2+0.003$ (\textbf{Theorem~\ref{theorem:one-pass-gen}}).

Our two-pass semi-streaming algorithm do not need the random order
assumption. We present a randomized two-pass algorithm with expected
approximation ratio $1/2+0.019$ against its internal random 
coin flips, for any bipartite graph and for any arrival order (\textbf{Theorem~\ref{l_theorem2}}).
We achieve the same approximation ratio with a more involved deterministic semi-streaming algorithm
(\textbf{Theorem~\ref{theorem:two-pass-det-bip}}). Last, we extend the 
latter algorithm to general graphs and we obtain an approximation ratio of 
$\frac{1}{2} + 0.0071$ (\textbf{Theorem~\ref{theorem:two-pass-det-gen}}). Figure~\ref{fig:results}
provides an overview of our algorithms.

\begin{figure}[ht]
\small
\begin{center}
 \begin{tabularx}{\textwidth}{|clX|}
\hline
Bipartite/General Graphs & Deterministic/Randomized & Approximation Factor \\
\hline
\multicolumn{3}{|l|}{\textbf{$1$ pass, uniform random order:}} \\
bipartite & deterministic & $\frac{1}{2} + 0.005$ (\textbf{Theorem~\ref{theorem:one-pass-bip}}) \\
general & deterministic & $\frac{1}{2} + 0.003$ (\textbf{Theorem~\ref{theorem:one-pass-gen}}) \\
\multicolumn{3}{|l|}{\textbf{$2$ passes, arbitrary order:}} \\
bipartite & randomized & $\frac{1}{2} + 0.019$ (\textbf{Theorem~\ref{l_theorem2}}) \\
bipartite & deterministic &  $\frac{1}{2} + 0.019$ (\textbf{Theorem~\ref{theorem:two-pass-det-bip}}) \\
general & deterministic & $\frac{1}{2} + 0.0071$ (\textbf{Theorem~\ref{theorem:two-pass-det-gen}}) \\
\hline
\end{tabularx}
\caption{Overview of our semi-streaming algorithms for maximum matching. \label{fig:results}}
\end{center}
\end{figure}

{\bf Techniques.}
The one-pass algorithms as well as the randomized 
two-pass algorithm each apply three times
the greedy matching algorithm on different and carefully chosen subgraphs. The deterministic two-pass 
algorithms are more complicated as they use subroutines that compute particular subsets of edges 
besides the greedy algorithm. There is a general idea that is common to all our algorithms that we
are going to explain for bipartite graphs:

If we had three passes 
at our disposal (see for instance Algorithm~2 in \cite{fkmsz05}), we could use one pass to build 
a maximal matching $M_0$ between the two sides $A$ and $B$ of the bipartition, a second pass to 
find a matching $M_1$ between the $A$ vertices matched in $M_0$ and the $B$ vertices that are free
with respect to $M_0$ whose combination with edges of $M_0$ forms paths of length 2. Finally, a third 
pass to find a matching $M_2$ between $B$ vertices matched in $M_0$ and $A$ vertices 
that are free with respect to $M_0$ whose combination with $M_0$ and $M_1$ forms paths of length 3 that 
can be used to augment the matching $M_0$. All our algorithms simulate these $3$ passes in less passes.

\textit{One-pass algorithm for random arrival order:} 
To simulate this with a single pass, we split 
the sequence of arrivals $[1,m]$ into three phases $[1,\alpha m]$, $(\alpha m, \beta m]$, 
and $(\beta m, m]$ for $0 < \alpha < \beta < 1$, and we build $M_0$ during the first phase, $M_1$ during the second phase,
and $M_2$ during the third phase. Of course, we see only a subset of the edges for each phase, but thanks to 
the random order arrival, these subsets are random, and, intuitively, we loose 
only a constant fraction in the sizes of the constructed matchings. As it turns out, the
intuition can be made rigorous, as long as the first matching $M_0$ is maximal or close to maximal.   
We observe that, if the greedy algorithm, executed on the entire sequence of edges, produces a 
matching that is not much better than a $1/2$ approximation of the optimal maximum matching, 
then that matching {\em is built early on}. More precisely (Lemma~\ref{lemma:exp}), if the greedy 
matching on the entire graphs is no better than a $1/2+\epsilon$ approximation, then 
after seeing a mere one third of the edges of the graph, the greedy matching is already a $1/2 -\epsilon$ approximation, so it is already close to maximal. 

\textit{Randomized two-pass algorithm for any arrival order:}
Assume a bipartite graph $(A,B,E)$ comprising a perfect matching. If $A'$ is a small random 
subset of $A$, then, regardless of the arrival order, the greedy algorithm that constructs a greedy matching between $A'$ and $B$
 (that is, the greedy algorithm restricted to the edges that have an endpoint in $A'$) will find a 
matching that is near-perfect, that is, almost every vertex of $A'$ is matched (see Theorem~\ref{theorem:subset-match} for a slightly more
general version of this statement). This property of the greedy algorithm may be of independent interest.
Then, in one pass we compute a greedy matching $M_0$ and also via the greedy algorithm independently and in parallel a matching $M_1$
between a subset $A' \subset A$ and the $B$ vertices. It turns out that $M_0 \cup M_1$ comprise some length $2$
paths that can be completed to $3$-augmenting paths by a third matching $M_2$ that we compute in the second pass.

\textit{Deterministic two-pass algorithm for any arrival order:}
Again, assume a bipartite graph $(A, B, E)$ comprising a perfect matching and some integer $\lambda$. 
Add now greedily edges $ab$ to a set $S$ if the degree of $a$ in $S$ is yet $0$, and the degree
of $b$ is smaller than $\lambda$. This algorithm computes an \textit{incomplete semi-matching} with
a degree limitation $\lambda$ on the $B$ nodes and is also used in \cite{kr13}. In the first pass, we run this algorithm in parallel to the
greedy matching algorithm for constructing $M_0$. $S$ replaces the computation of $M_1$, and we will see that 
there are length $2$ paths in $M_0 \cup S$ that can be completed to $3$-augmenting paths in the second pass via a further greedy matching $M_2$.
  

{\bf Extension to general graphs.} The deterministic one-pass algorithm for bipartite graphs
and the deterministic two-pass algorithm for bipartite graphs both extend to general graphs.
When searching for augmenting paths in general graphs, algorithms
have to cope with the fact that a candidate edge for an augmenting path may form an undesired
triangle with the edge to augment and an optimal edge. In this case, the candidate edge 
can block the entire augmenting path. McGregor \cite{m05} overcomes this problem by 
repeatedly sampling bipartite graphs from the general graph. Such a strategy, however, is not necessary
for our one-pass algorithm. Since the input sequence is in uniform random order, 
we can show that undesired triangles simply do not appear \textit{too often} allowing our techniques to still work.
For our deterministic two-pass algorithm, a combinatorial argument is used to bound the 
number of those \textit{bad} triangles. 

{\bf Conference Version.} This work builds on the article \cite{kmm12} that was presented at the 
15th International Workshop on Approximation Algorithms for Combinatorial Optimization Problems (APPROX 2012).
Besides a more detailed presentation of the results of \cite{kmm12}, the extensions of the algorithms for
bipartite graphs to general graphs are discussed. 

{\bf Outline.} We start our presentation with notations and definitions in Section~\ref{sec:prelim}. 
In Section~\ref{section:matching-3-pass}, we discuss a well-known result that is reused in all following sections. 
We point out how the Greedy matching algorithm can be used in $3$ passes to obtain an approximation ratio 
strictly larger than $1/2$. 
In Section~\ref{section:matching-one-pass}, we discuss the one-pass algorithm for bipartite graphs
and its extension to general graphs. Then, in Section~\ref{section:matching-two-pass-randomized} we
present our randomized two-pass algorithm for bipartite graphs. Finally, we conclude with the
discussion of our deterministic two-pass algorithms for bipartite graphs and general graphs in 
Section~\ref{section:matching-two-pass-det}.

\section{Preliminaries} \label{sec:prelim}

Let $G=(V, E)$ be a graph with vertex set $V$ and edge set $E$. If $G$ is bipartite with bipartitions $A$ and $B$ then 
we write $G = (A, B, E)$ and we denote $V = A \cup B$. Let $n=|V|$ and $m=|E|$.
For an edge $e \in E$ with end points $u,v\in V$, we denote
$e$ by $uv$. For a subset of edges $S \subseteq E$ and a vertex $v \in V$, we write $\deg_S(v)$
for the degree of $v$ in $S$, meaning the number of edges in $S$ that have $v$ as one of its endpoints.

We define now matchings, maximum matchings and maximal matchings.

\begin{definition}[Matching]
A {\em matching} in a graph $G = (V, E)$ is a subset of edges $M \subseteq E$ such that 
$\forall v \in V: \deg_M(v) \le 1$.  
A {\em maximum matching} $M^*$ is a matching such that for 
any other matching $M': |M^*| \ge |M'|$. 
A {\em maximal matching} $M$ is a matching that is 
inclusion-wise maximal, a.e. $\forall e \in E \setminus M: M \cup \{ e\}$ is not a matching.
\end{definition}

The \mbm~problem consists of computing a maximum matching in a bipartite graph and we abbreviate it
by \mbmshort. 

The \mm~problem consists of computing a maximum matching in a general graph and we abbreviate it
by \mmshort.

For a subset of edges $F \subseteq E$, we denote by $\OPT(F)$ a maximum matching in the graph  
$G$ restricted to edges $F$. 
We may write $\OPT(G)$ for $\OPT(E)$, and $M^*$ for $\OPT(G)$. 
For a set of vertices $S$ and a set of edges $F$, let $S(F)$
be the subset of vertices of $S$ covered by $F$. 
Furthermore, we use the abbreviation 
$\overline{S(F)} := S \setminus S(F)$.
For $S \subseteq V$, we write $\OPT(S)$ for $\OPT(G|_S)$, that is a maximum matching
in the subgraph of $G$ induced by vertices $S$. In case of bipartite graphs, for 
$S_A \subseteq A$ and $S_B \subseteq B$ we write $\OPT(S_A, S_B)$ for $\OPT(G|_{S_A \cup S_B})$.
Moreover, for two sets $S_1, S_2$ we denote by $S_1 \oplus S_2$ 
the symmetric difference $(S_1 \setminus S_2) \cup (S_2 \setminus S_1)$ of the two sets.

A standard technique to increase the size of matchings is to search for {\em augmenting paths}. We 
define augmenting paths as follows.

\begin{definition}[Augmenting Path]
 Let $p  \ge 3$ be an odd integer. Then a \textit{length $p$ augmenting path} with respect to a matching $M$ in a graph $G=(V, E)$ 
is a path $P = (v_1, \dots, v_{p+1})$ such that $v_1, v_{p+1} \notin V(M)$ and for $i \le 1/2(p-1): v_{2i}v_{2i+1} \in M$, 
and $v_{2i-1}v_{2i} \notin M$. 
\end{definition}




An augmenting path of length $p$ ($p \ge 3, p$ odd) with respect to a matching $M$ in a graph $G=(V, E)$ is a path that
starts and ends at nodes that are not matched in $M$. We call such nodes {\em free} nodes. All internal nodes of the path
are matched in $M$, and we call these nodes {\em matched} nodes. The path alternates between edges outside $M$ and edges of $M$. 
Removing from $M$ the edges of the augmenting path that are also in $M$ and inserting into $M$ the edges outside $M$ increases
the size of $M$ by $1$. 

The input graph $G$ is given as a graph stream, i.e. as a sequence of edges arriving one by one in 
some order. Let $\Pi(G)$ be the set of all edge sequences of $G$. 
An input stream for our streaming algorithms is then an edge sequence $\pi \in \Pi(G)$. 
We write $\pi[i]$ for the $i$-th 
edge of $\pi$, and $\pi[i, j]$ for the subsequence $\pi[i] \pi[i+1] \dots \pi[j]$.
In this notation, a round bracket excludes the smallest or respectively largest
element: $\pi(i, j] = \pi[i + 1, j]$, and $\pi[i, j) = \pi[i, j-1]$.
If $i, j$ are real, $\pi[i, j] := \pi[\lfloor i \rfloor, \lfloor j \rfloor ]$, and
$\pi[i] := \pi[ \lfloor i \rfloor ]$.
Given a subset $S \subseteq V$, $\pi|_S$ is the largest subsequence of $\pi$ 
such that all edges in $\pi|_S$ are among vertices in $S$.

\begin{definition}[Semi-streaming Algorithm]
A $p(n)$-pass {\em semi-streaming algorithm} \textbf{S} on input graph $G$ with update time $t(n)$ is an
algorithm such that, for every input stream $\pi \in \Pi(G)$: 
\begin{enumerate}
\itemsep1pt
\parsep1pt
 \item \textbf{S} performs at most $p(n)$ passes on stream $\pi$, 
 \item \textbf{S} maintains a random access memory of size $\Order(n \polylog n)$, 
 \item \textbf{S} has running time $\Order(t(n))$ between two consecutive read operations from the stream.
\end{enumerate}
Furthermore, preprocessing time (the time before the first read operation) and postprocessing time (the time after 
the last read operation and the output of the result) is $\Order(t(n))$. We assume that read operations on 
any stream require constant time.
\end{definition}

We say that an algorithm $\mathbf{A}$ computes a \emph{$c$-approximation to the maximum matching problem}
if $\mathbf{A}$ outputs a matching $M$ such that $|M|\geq c \cdot |\OPT(G)|$.
We consider two potential sources of randomness: from the algorithm and from the arrival order.
Nevertheless, we will always consider worst case against the graph.
For each situation, we relax the notion of $c$-approximation so that the expected approximation ratio is $c$,
that is $\Exp |M|\geq c \cdot |\OPT(G)|$
where the expectation can be taken either over  the internal random coins of the algorithm,
or over all possible arrival orders. 

The Greedy matching algorithm is illustrated in Algorithm~\ref{algo:greedy-matching}. It is
easy to see that this algorithm can be seen as a semi-streaming algorithm with approximation
factor $1/2$ and update time $\Order(1)$.

\begin{algorithm}[H]\small
\caption{The $\Greedy$ Matching Algorithm \label{algo:greedy-matching}}
\begin{algorithmic}[1]
 \STATE $M \gets \varnothing$
 \WHILE{edge stream not empty}
  \STATE $e=v_1 v_2 \gets $ next edge in stream
  \STATE \textbf{if} {$\{v_1, v_2 \} \cap V(M) = \varnothing$} \textbf{then} $M \gets M \cup \{e \}$ \textbf{end if}
 \ENDWHILE
 \RETURN $M$
\end{algorithmic}
\end{algorithm}


\section{Three-pass Semi-Streaming Algorithm for Bipartite Graphs on Adversarial Order} \label{section:matching-3-pass}


To improve on the Greedy matching algorithm with three passes, a simple strategy is to, firstly, compute
a maximal matching $M_G$ in one pass, and then use the second and the third pass to search for $3$-augmenting paths
to augment $M_G$.

Suppose that $M_G$ is close to a $1/2$-approximation.
Then almost all edges of $M_G$ are {\em $3$-augmentable}. We say that an edge $e \in M_G$ is
$3$-augmentable if the removal of $e$ from $M$ allows the insertion of two edges $f,g \in M^* \setminus M$ into $M$. More formally,
the following lemma holds.

\begin{lemma}
\label{lemma:augmentable_edges}
Let $\epsilon \ge 0$. 
Let $M$ be a maximal matching of $G$ st. $|M| \le (\frac{1}{2} + \epsilon) |M^*|$. 
Then $M$ contains at least $ (\frac{1}{2} - 3\epsilon) |M^*|$ $\mbox{$3$-augmentable}$ edges. 
\end{lemma}
\begin{proof}
The proof is folklore. Let $k_i$ denote the number of paths of length $i$ in $M \oplus M^*$. 
Since $M^*$ is maximum, it has no augmenting path, so all odd length paths are augmenting 
paths of $M$. Since $M$ is maximal, there are no augmenting paths of length 1, so  $k_1=0$. 
Every even length path and every cycle has an equal number of edges from $M$ and from $M^*$. 
A path of length $2i+1$ has $i$ edges from $M$ and $i+1$ edges from $M^*$. 
$$ |M^*|-|M|
=\sum_{i\geq 1 } k_{2i+1} 
\leq k_3+\sum_{i\geq 2} \frac{1}{2}i k_{2i+1}
= \frac{1}{2}k_3+\frac{1}{2}\sum_{i\geq 1} i k_{2i+1}
\leq  \frac{1}{2}k_3+\frac{1}{2} |M|.
$$
Thus, using our assumption on $|M|$, $k_3\geq 2 |M^*| - 3 |M| \geq 2|M^*|- (\frac{3}{2}+3\epsilon )|M^*|,$
implying the Lemma. \qed
\end{proof}

\begin{figure}[ht]
\begin{center}
$\quad$ \includegraphics[height=5.5cm, bb=120 330 475 515]{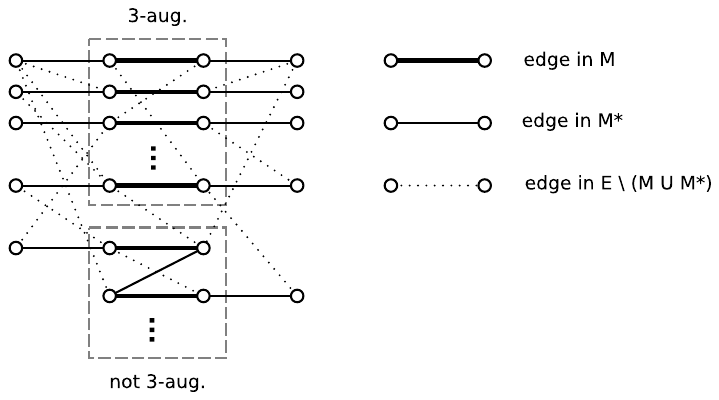}
\caption[Small Maximal Matchings admit many length $3$-Augmenting Paths.]{ Illustration of Lemma~\ref{lemma:augmentable_edges}. 
If $|M| \le (1/2 + \epsilon)|M^*|$, then at least $ (\frac{1}{2} - 3\epsilon) |M^*|$ edges of $M$ are
$3$-augmentable. \label{fig:many-length-3-aug-paths} }
\end{center}
\end{figure}

We search for $3$-augmenting paths as follows. Firstly, we compute a maximal matching $M_L$ via the Greedy algorithm
between the $A$ vertices that are matched in $M_G$ and the free $B$ vertices. Under the assumption that $M_G$ is close to a $1/2$ approximation,
most of the edges of $M_G$ are $3$-augmentable.  There exists hence a large matching, and since $M_L$ is a maximal matching, $M_L$ will 
be at least of size $1/2$ times the number of $3$-augmentable edges. Edges from $M_L$ will serve as the start of 
length $3$-augmenting paths. Then in the third pass, 
we compute another maximal matching $M_R$ in order to complete $3$-augmenting paths
with the edges of $M_G$ and $M_L$. This algorithm is stated in Algorithm~\ref{algo:3-pass-matching}, and illustrated in 
Figure~\ref{fig:algo-3-passes}. This idea was already used in \cite{fkmsz05}. The authors present there an $\Order( (\log \frac{1}{\epsilon})/\epsilon)$-pass 
semi-streaming algorithm that computes a $2/3 - \epsilon$ approximation to the maximum bipartite matching problem. An analysis for Algorithm~\ref{algo:3-pass-matching}
can be derived from their work.

\begin{algorithm}[H]\small
\caption{Three-pass Bipartite Matching Algorithm \label{algo:3-pass-matching}}
\begin{algorithmic}[1]
 \REQUIRE The input stream $\pi$ is an edge stream of a bipartite graph $G = (A, B, E)$
 \STATE $M_G, M_L, M_R \gets \varnothing$
 \STATE \textbf{1\textsuperscript{st} pass:} $M_G \gets \Greedy(\pi)$
 \STATE $G_L \gets $ complete graph between $A(M_G)$ and $B \setminus B(M_G)$
 \STATE \textbf{2\textsuperscript{nd} pass:} $M_L \gets \Greedy(\pi \cap G_L)$
 \STATE $G_R \gets $ complete graph between $\{b \in B(M_G) \, : \, A(M_G(b)) \in A(M_L) \}$ and $A \setminus A(M_G)$
 \STATE \textbf{3\textsuperscript{rd} pass:} $M_R \gets \Greedy(\pi \cap G_R)$
 \RETURN maximum matching in $M_G \cup M_L \cup M_R$
\end{algorithmic}
\end{algorithm}

\begin{figure}[ht]
\begin{center}
\includegraphics[height=5cm, bb=0 0 350 180]{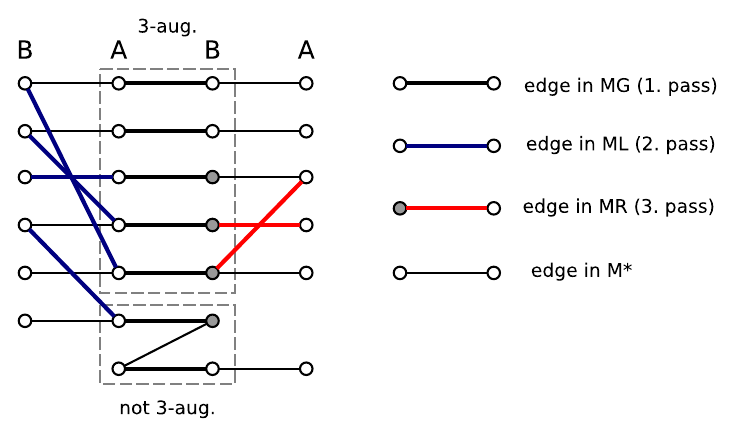}
\caption[Three-pass Bipartite Matching Algorithm]{ Illustration of Algorithm~\ref{algo:3-pass-matching}. The graph
contains a perfect matching of size $13$. In the first pass, $M_G$ is computed and it has size $7$. This is close to a $1/2$
approximation and by Lemma~\ref{lemma:augmentable_edges}, $M$ has many (here $5$) $3$-augmentable edges. There exists hence a matching of size
at least $5$ between $A(M_G)$ and the free $B$ vertices. Since $M_L$ is maximal, it is of size at least $5/2$ (here $4$).
Then, a maximal matching is computed between the solid vertices, which are the $B$ vertices of edges of $M_G$ that may potentially 
be completed to $3$-augmenting paths, and the free $A$ vertices. In this example, two $3$-augmenting paths were found.
\label{fig:algo-3-passes} }
\end{center}
\end{figure}



\section{One-pass Matching Algorithm on Random Order} \label{section:matching-one-pass}

We discuss now, how the $3$-pass algorithm from the previous section can
be simulated with a single pass if the input is in random order. First, we present in 
Subsection~\ref{section:convergence-greedy} a lemma about the convergence of the Greedy 
matching algorithm if the input is in random order. This lemma is the main ingredient for
our one-pass algorithms. Then, in Subsection~\ref{section:one-pass-bipartite} we discuss 
our one-pass algorithm on random order for bipartite graphs, and we extend it to general 
graphs in Subsection~\ref{section:one-pass-general}.

\subsection{A Lemma on the Convergence of the Greedy Algorithm} \label{section:convergence-greedy}
We identify a property about the convergence of the Greedy algorithm that is required 
for the construction of our one-pass algorithms on random order. We show that if in expectation
over all input edges sequences the Greedy algorithm computes a matching that is close to a $1/2$ 
approximation, then Greedy builds this matching early on, or in other words, the Greedy algorithm
converges quickly, see Lemma~\ref{lemma:exp}.

\begin{lemma}\label{lemma:exp}
If $\Exp_\pi |\Greedy(\pi)| \leq (\frac{1}{2} + \epsilon) |M^*|$ for some $0 < \epsilon < 1/2$,
then for any $0 < \alpha \le 1$,
\begin{equation*}
\Exp_{\pi} |\Greedy(\pi[1, \alpha m])| \ge |M^*|(\frac{1}{2}-(\frac{1}{\alpha}-2) \epsilon) . 
\end{equation*}
\end{lemma} 

\begin{proof}
Let $M_0 = \Greedy(\pi[1, \alpha m])$. Rather than directly analyzing the number of edges 
$|M_0|$, we analyze the number of vertices matched by $M_0$, which is equivalent since 
$|V(M_0)|=2 (|M_0|).$

Fix an edge $e= ab$ of $M^*$. Either $e\in M_0$, or at least one of $a,b$ is matched by 
$M_0$, or neither $a$ nor $b$ are matched. Summing over all $e\in M^*$ gives
$$|V(M_0)|\geq 2|M^*\cap M_0|+|M^*\setminus M_0|- 
\sum_{e=ab \in M^*} \chi[a \mbox{ and } b \notin V(M_0)] ,$$
where $\chi[X] = 1$ if the event $X$ happens, otherwise $\chi[X] = 0$.
We show in Lemma~\ref{lemma:correlation} that 
\begin{equation}
\Pr[a \mbox{ and } b \notin V(M_0) ]\leq (\frac{1}{\alpha}-1) \Pr [ e\in M_0] . \label{eqn:493}
\end{equation}

Taking expectations and using Inequality~\ref{eqn:493},
\begin{eqnarray*}
\Exp_{\pi}(|V(M_0)|) &\geq & 2\Exp_{\pi} |M^*\cap M_0|+ \Exp_{\pi} |M^*\setminus M_0|- 
(\frac{1}{\alpha }-1)\Exp_{\pi} |M^*\cap M_0| \\
& = & |M^*| - (\frac{1}{\alpha }-2)\Exp_{\pi} |M^*\cap M_0|  .
\end{eqnarray*}
We will show in Lemma~\ref{lemma:max_opt_edges} that for a maximum matching $M^*$ and any maximal matching $M_G$, 
we have $|M_G \cap M^*| \le 2(|M_G| - 1/2 |M^*|)$.
Using this, and since $M_0$ is just a subset of the edges of $M_G$, we obtain
by linearity of expectation
$$\Exp_{\pi} |M^*\cap M_0|\leq \Exp_{\pi} |M^* \cap M_G|
\leq 2(\Exp_{\pi}|M_G|-\frac{1}{2}|M^*|) \leq 2\epsilon |M^*|.$$
Combining gives the Lemma.  \qed
\end{proof}

We now prove Lemma~\ref{lemma:correlation} that was used in the proof of Lemma~\ref{lemma:exp}.

\begin{lemma}\label{lemma:correlation}
Suppose that $\Exp_\pi |\Greedy(\pi)| \leq (\frac{1}{2} + \epsilon) |M^*|$ for some $0 < \epsilon < 1/2$.
Let $M_0 = \Greedy(\pi[1, \alpha m])$ for some $0 < \alpha \le 1/2$. Then:
\begin{equation*}
\forall e=ab \in E: \Pr[a \mbox{ and } b \notin V(M_0) ]\leq (\frac{1}{\alpha }-1) \Pr [ e \in M_0] . 
\end{equation*}

\end{lemma}
\begin{proof}
Observe: $\Pr[a \mbox{ and } b \notin V(M_0)]+\Pr [ e \in M_0]= \Pr[a \mbox{ and } b \notin V(M_0 \setminus \{e \})],$
because the two events on the left hand side are disjoint and their union is the event on the right hand side.

Consider the following probabilistic argument. 
Take the execution for a particular ordering $\pi$. Assume that $a \mbox{ and } b \notin V(M_0 \setminus \{e \})$  
and let $t$ be the arrival time of $e$. If we modify the ordering 
by changing the arrival time of $e$ to some 
time $t'\leq t$, then we still have $a \mbox{ and } b \notin V(M_0 \setminus \{e \})$. 
More formally, we define a map $f$ from the uniform distribution on all orderings to the uniform distribution on all orderings such that $e\in\pi [1,\alpha m]$: if $e\in \pi[1,\alpha m]$ then $f(\pi)=\pi$ and otherwise $f(\pi)$ is the permutation obtained from $\pi$ by removing $e$ and re-inserting it at a position picked uniformly at random in $[1,\alpha m]$. Thus,
$$\Pr [ a \mbox{ and } b \notin V(M_0 \setminus \{e \})]\leq 
\Pr [ a \mbox{ and } b \notin V(M_0 \setminus \{e \})|   e\in\pi[1,\alpha m] ].$$
Now, the right-hand side equals $\Pr [e\in M_0 | e\in \pi[1,\alpha m]]$, which simplifies into
$\Pr [ e\in M_0] /\Pr [e\in \pi[1,\alpha m]]$ since $e$ can only be in $M_0$ if it is one of the first $\alpha m$ arrivals.
Then we conclude the Lemma by the random order assumption $\Pr [e\in \pi[1,\alpha m]] = \alpha$. \qed
\end{proof}

Lemma~\ref{lemma:max_opt_edges} shows that an optimal matching and a maximal matching that is far 
from this optimal matching in size do not have many edges in common.\vspace*{-1mm}

\begin{lemma} 
\label{lemma:max_opt_edges}
Let $M$ be a maximal matching of a graph $G$. Then 
\begin{equation*}
|M \cap M^*| \le 2(|M|-\frac{1}{2}|M^*|).
\end{equation*}
\end{lemma}

\begin{proof}
This is a piece of elementary combinatorics. Since $M$ is a maximal 
matching, for every edge $e$ of $M^*\setminus M$, at least one of the 
two endpoints of $e$ is matched in $M\setminus M^*$, and so 
$|M \setminus M^*|\geq (1/2)  |M^* \setminus M|$. We have 
$ |M^* \setminus M| = |M^*|-|M^*\cap M|$. Combining gives
$$ |M \cap M^*| = |M| - |M \setminus M^*| \le |M| - \frac{1}{2} |M^* \setminus M| = |M| - \frac{1}{2} (|M^*| - |M^* \cap M|)$$
which implies the Lemma. \qed
\end{proof}

\subsection{Bipartite Graphs} \label{section:one-pass-bipartite}
\subsubsection{Algorithm}
We simulate the $3$-pass algorithm, Algorithm~\ref{algo:3-pass-matching}, in one pass as follows. 
We split the input graph stream $\pi \in \Pi(G)$ into three 
phases $\pi[1,\alpha m]$, $\pi(\alpha m, \beta m]$, and $\pi(\beta m, m]$ (for $0 < \alpha < \beta < 1$), and we build a 
matching in each phase. $M_0$ is built during the first phase and corresponds to matching $M_G$ of our $3$-pass algorithm. $M_1$ is built in the second phase and $M_2$ in the
third, and they correspond to $M_L$ and $M_R$ of our $3$-pass algorithm, respectively.
Assume that $\Greedy$ performs badly on the input graph $G$. Lemma~\ref{lemma:augmentable_edges} tells us that almost all 
of the edges of $M_0$ are $3$-augmentable. To find $3$-augmenting paths, in the next part of the 
stream, we run $\Greedy$ to compute a matching $M_1$ between $B(M_0)$ and $\overline{A(M_0)}$.
The edges in $M_1$ serve as one of the edges of $3$-augmenting paths
(from the $B$-side of $M_0$). 
In Lemma~\ref{l_number_of_right_wings}, we show that we find a constant fraction of those. 
In the last part of the stream, again by the help of $\Greedy$, we compute a matching
$M_2$ that completes the $3$-augmenting paths. Lemma~\ref{l_left_wings} shows 
that by this strategy we find many $3$-augmenting paths. 
Then, either a simple $\Greedy$ matching performs well on $G$, or else we can find many $3$-augmenting paths and
use them to improve $M_0$, see the main 
theorem, Theorem~\ref{theorem:one-pass-bip}, whose proof is deferred to the end of this section.
An illustration is provided in Figure~\ref{l_fig_0}. 




\begin{algorithm} 
\caption{One-pass Bipartite Matching on Random Order \label{l_algo1}}
\begin{algorithmic}[1]
\STATE $\alpha \gets 0.4312, \beta \gets 0.7595$
\STATE $M_G \gets \Greedy(\pi)$
\STATE $M_0 \gets \Greedy(\pi[1,\alpha m])$, matching obtained by $\Greedy$ on the first $\lfloor \alpha m \rfloor$~edges
\STATE $F_1 \gets $ complete bipartite graph between $B(M_0)$ and $\overline{A(M_0)}$
\STATE $M_1 \gets \Greedy(F_1 \cap\pi(\alpha m, \beta m])$,  matching obtained by Greedy on edges $\lfloor \alpha m\rfloor +1$ through $\beta m$ that intersect $F_1$
\STATE   $A' \gets \{ a \in A \, | \, \exists b\in B(M_1): ab \in M_0 \}$
\STATE   $F_2 \gets $  complete bipartite graph between $A'$ and $ \overline{B(M_0)}$
\STATE $M_2 \gets \Greedy(F_2 \cap \pi(\beta m, m])$,  matching obtained by $\Greedy$ on edges $\lfloor \beta m\rfloor +1$ through $m$ that intersect $F_2$
\STATE  $M\gets $  matching obtained from $M_0$ augmented by $M_1 \cup M_2$
\RETURN larger of the two matchings  $M_G$ and $M$
\end{algorithmic}
\end{algorithm}

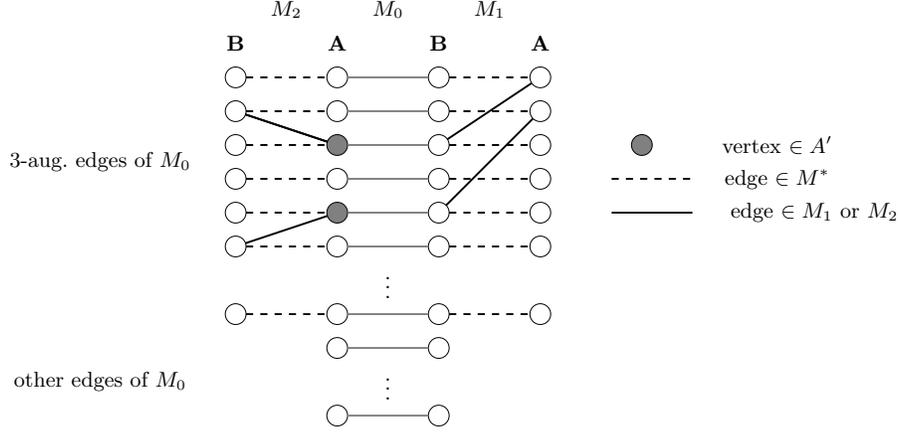
\begin{figure}[ht!]
\centerline{\scalebox{0.9}{
\begin{tikzpicture}[node distance = 1.5 cm]
    \GraphInit[vstyle=normal]
\draw(-2.0, 8.75) node{$3$-aug. edges of $M_0$};
\draw(-2.0, 5.5) node{other edges of $M_0$};
\draw(2.25, 11) node{\textbf{$M_0$}};
\draw(3.75, 11) node{\textbf{$M_1$}};
\draw(0.75, 11) node{\textbf{$M_2$}};
\draw(0, 10.5) node{\textbf{B}};
\draw(1.5, 10.5) node{\textbf{A}};
\draw(3, 10.5) node{\textbf{B}};
\draw(4.5, 10.5) node{\textbf{A}};
\draw(2.25, 7) node{$\vdots$};
\draw(2.25, 5.5) node{$\vdots$};
\draw(8, 9) node{vertex $\in A'$};
\draw(8.55, 8) node{edge $\in M_1$ or $M_2$};
\draw(8, 8.5) node{edge $\in M^*$};
\tikzset{VertexStyle/.style = {shape = circle, color = black, fill = white,minimum size = 0.1cm, draw}}
\Vertex[x=0, y=10, L=$$]{d1} \Vertex[x=1.5, y=10, L=$$]{a1} \Vertex[x=3, y=10, L=$$]{b1} 
\Vertex[x=0, y=9.5, L=$$]{d2} 				 \Vertex[x=3, y=9.5, L=$$]{b2} \Vertex[x=4.5, y=9.5, L=$$]{c2}
\Vertex[x=0, y=9, L=$$]{d3}  \Vertex[x=3, y=9, L=$$]{b3} \Vertex[x=4.5, y=9, L=$$]{c3}                
\Vertex[x=0, y=8.5, L=$$]{d4} \Vertex[x=1.5, y=8.5, L=$$]{a4} \Vertex[x=3, y=8.5, L=$$]{b4}                 
\Vertex[x=0, y=8, L=$$]{d5}  \Vertex[x=3, y=8, L=$$]{b5} \Vertex[x=4.5, y=8, L=$$]{c5}                
\Vertex[x=0, y=7.5, L=$$]{d6} 				 \Vertex[x=3, y=7.5, L=$$]{b6} \Vertex[x=4.5, y=7.5, L=$$]{c6}
\Vertex[x=0, y=6.5, L=$$]{d7} \Vertex[x=1.5, y=6.5, L=$$]{a7} \Vertex[x=3, y=6.5, L=$$]{b7} \Vertex[x=4.5, y=6.5, L=$$]{c7}
\Vertex[x=1.5, y=6, L=$$]{a8} \Vertex[x=3, y=6, L=$$]{b8} 
\Vertex[x=1.5, y=5, L=$$]{a9} \Vertex[x=3, y=5, L=$$]{b9} 
\Vertex[x=1.5, y=9.5, L=$$]{a2}
\Vertex[x=1.5, y=7.5, L=$$]{a6}
\Vertex[x=4.5, y=10, L=$$]{c1}
\Vertex[x=4.5, y=8.5, L=$$]{c4}
%
\tikzset{VertexStyle/.style = {shape = circle, color = black, fill = gray, minimum size = 0.1cm, draw}}
\Vertex[x=1.5, y=9, L=$$]{a3}
\Vertex[x=1.5, y=8, L=$$]{a5}
\Vertex[x=6.0, y=9, L=$$]{e1} 
\tikzset{VertexStyle/.style = {shape = circle, color = white, fill = white, minimum size = 0.1cm, draw}}
\Vertex[x=6.9, y=8.5, L=$$]{f3}
\Vertex[x=5.4, y=8.5, L=$$]{e3}
\Vertex[x=5.4, y=8, L=$$]{f4}
\Vertex[x=6.9, y=8, L=$$]{e4}
\SetUpEdge[style={solid}, color=gray]
\Edge(a1)(b1)
\Edge(a2)(b2)
\Edge(a3)(b3)
\Edge(a4)(b4)
\Edge(a5)(b5)
\Edge(a6)(b6)
\Edge(a7)(b7)
\Edge(a8)(b8)
\Edge(a9)(b9)
\SetUpEdge[style={dashed}, color=black]
\Edge(d1)(a1) \Edge(b1)(c1)
\Edge(d2)(a2) \Edge(b2)(c2)
\Edge(d3)(a3) \Edge(b3)(c3)
\Edge(d4)(a4) \Edge(b4)(c4)
\Edge(d5)(a5) \Edge(b5)(c5)
\Edge(d6)(a6) \Edge(b6)(c6)
\Edge(d7)(a7) \Edge(b7)(c7)
\Edge(e3)(f3)
\SetUpEdge[style={dotted}, color=black]
\SetUpEdge[style={thick}, color=black]
\Edge(b5)(c2)
\Edge(a3)(d2)
\Edge(a3)(d2)
\Edge(a5)(d6)
\Edge(f4)(e4)
\Edge(c1)(b3)
\end{tikzpicture} }}
\caption[Analysis of One-pass Bipartite Matching Algorithm]{Illustration of Algorithm~\ref{l_algo1}. Note that every edge of $M_2$ completes a $3$-augmenting path consisting of one edge of $M_1$ (on the right hand side of the picture) followed by one edge of $M_0$ (center)
followed by one edge of $M_2$ (on the left hand side of the picture). \label{l_fig_0} } 
\end{figure}

Observe that our algorithm only uses memory space $\Order(n\log n)$.
Indeed, the subsets $F_1$ and $F_2$ can be compactly represented by two
$n$-bit arrays,
and checking if an edge of $\pi$ belongs to one of them can be done within time
$\Order(1)$ from that compact representation.

\hspace{0.3cm}

\subsubsection{Analysis}
We use the notations of Algorithm~\ref{l_algo1}. Consider $\alpha$ and $\beta$
as variables with $0 \le \alpha \le \frac{1}{2} < \beta < 1$.

\begin{lemma}
\label{l_number_of_right_wings}
Assume that $\Exp_\pi |M_G|\leq (\frac{1}{2} + \epsilon) |M^*|$. 
Then the expected size of a maximum matching between the vertices of $A$ 
left unmatched by $M_0$ and the vertices of $B$ matched by $M_0$ can be bounded below as follows:
$$
 \Exp_{\pi}  |\OPT( \overline{A(M_0)}, B(M_0) )| \ge |M^*|(\frac{1}{2}-(\frac{1}{\alpha}+2) \epsilon) .
$$
\end{lemma}

\begin{proof}
The size of a maximum matching between $\overline{A(M_0)}$ and $ B(M_0)$ 
is at least the number of augmenting paths of length 3 in $M_0\oplus M^*$. 
By 
Lemma~\ref{lemma:augmentable_edges}, in expectation, the number of augmenting paths of 
length $3$ in $ M_G  \oplus M^*$ is at least $ (\frac{1}{2} - 3 \epsilon)|M^*|.$
All of those are augmenting paths of length $3$ in $M_0\oplus M^*$, 
except for at most $|M_G| - |M_0|$. Hence, in expectation, $M_0$ contains
$(\frac{1}{2} - 3 \epsilon)|M^*| - (\Exp_{\pi} |M_G| - \Exp_{\pi} |M_0|)$
$3$-augmentable edges. Lemma~\ref{lemma:exp} applied to $M_0$ concludes the proof. \qed
\end{proof}

\begin{lemma}
\label{l_exp_right_wings}
$ 
\Exp_{\pi} |M_1| \ge \frac{1}{2}(\beta - \alpha )( \Exp_{\pi} |\OPT(\overline{A(M_0)}, B(M_0))|-\frac{1}{1-\alpha}).
$
\end{lemma}
\begin{proof}
Since $\Greedy$ computes a maximal matching which is at least half the
size of a maximum matching, 
\begin{equation*}
\Exp_{\pi} |M_1|\geq \frac{1}{2} \Exp_{\pi}  |\OPT(\overline{A(M_0)}, B(M_0))\cap \pi(\alpha m, \beta m]|. 
\end{equation*}

By independence of $M_0$ and the ordering within $(\alpha m, m]$, we see that even if we condition on 
$M_0$, we still have that $\pi(\alpha m,\beta m]$ is a random uniform subset of $\pi(\alpha m,m ]$. Thus:
\begin{eqnarray*}
\Exp_{\pi} |\OPT(\overline{A(M_0)}, B(M_0))\cap \pi(\alpha m, \beta m]| = \quad\quad\quad\\
\quad\quad\quad\quad\tfrac{\beta -\alpha }{1-\alpha } \Exp_{\pi} |\OPT(\overline{A(M_0)}, B(M_0))\cap \pi(\alpha m, m]|. 
\end{eqnarray*}

We use a probabilistic argument similar to but slightly more complicated than in the proof 
of Lemma~\ref{lemma:correlation}. We define a map $f$ from the uniform distribution on all 
orderings to the uniform distribution on all orderings such that $e\in\pi (\alpha m, m] $:
 if $e\in \pi (\alpha m, m] $ then $f(\pi)=\pi$ 
 and otherwise $f(\pi)$ is the permutation obtained from $\pi$ by removing $e$ and re-inserting it at a position picked uniformly at random in $ (\alpha m, m] $; 
 in the latter case, if this causes an edge $f=a'b'$, previously arriving at time $\lfloor \alpha m\rfloor +1$, to now arrive at time $\lfloor \alpha m\rfloor$ and to be added to $M_0$, 
 we define $M'_0=M_0\setminus \{ f\}$; in all other cases we define $M'_0=M_0$. 
 Thus, if in $\pi$ we have $e\in \OPT(\overline{A(M_0)}, B(M_0))$,
  then in $f(\pi)$ we have $e \in \OPT(\overline{A(M'_0)}, B(M'_0))$. 
Since the distribution of $f(\pi)$ is uniform conditioned on $e\in \pi(\alpha m, m] $:
\begin{eqnarray*}
\frac{ \Pr [ e\in \OPT(\overline{A(M'_0)}, B(M'_0)) \hbox{ and } e\in \pi(\alpha m,m] ] }{ \Pr [ e\in \pi(\alpha m, m] ] }
\geq
 \Pr [e\in \OPT(\overline{A(M_0)}, B(M_0))], 
\end{eqnarray*}
Using $\Pr [ e\in \pi(\alpha m, m] ] = 1-\alpha$ and summing over $e$:
\begin{eqnarray*}
\Exp_{\pi}  |\OPT(\overline{A(M'_0)}, B(M'_0))\cap \pi(\alpha m, m]| \geq 
 (1-\alpha) \Exp_{\pi} |\OPT(\overline{A(M_0)}, B(M_0)) |. 
\end{eqnarray*}

Since $M'_0$ and $M_0$ differ by at most one edge, 
$|\OPT(\overline{A(M_0)}, B(M_0))|\geq |\OPT(\overline{A(M'_0)}, B(M'_0))|-1$, 
and the Lemma follows.  \qed
\end{proof}

\begin{lemma}
\label{l_matching_size_left_wings}
Assume that $\Exp_\pi |M_G|\leq (\frac{1}{2} + \epsilon) |M^*|$. Then:
\begin{eqnarray*}
\Exp_{\pi} |\OPT(A', \overline{B(M_0)}| \ge  \Exp_{\pi} |M_1| - 4 \epsilon |M^*|. 
\end{eqnarray*}

\end{lemma}
\begin{proof} $ |\OPT(A', \overline{B(M_0)}| $ is at least $|M_1|$ minus the 
number of edges of $M_0$ that are not 3-augmentable. Since $M_0$ is a subset 
of $M_G$, the latter term is bounded by the number of edges of $M_G$ that are 
not 3-augmentable, which by Lemma~\ref{lemma:augmentable_edges} is in expectation at most 
$(\frac{1}{2} + \epsilon) |M^*| - (\frac{1}{2} - 3\epsilon)|M^*| = 4 \epsilon |M^*|$.  \qed
\end{proof}

\begin{lemma}
\label{l_left_wings}
$\displaystyle\Exp_{\pi} |M_2| \ge \frac{1}{2} ((1 - \beta )\Exp_{\pi} |\OPT(A', \overline{B(M_0)})|-1).
 $
 \end{lemma}
\begin{proof}
Since $\Greedy$ computes a maximal matching which is at least half the
size of a maximum matching, 
\begin{equation*}
\Exp_{\pi} |M_2|\geq \frac{1}{2} \Exp_{\pi}  |\OPT(A', \overline{B(M_0)})\cap \pi(\beta m,  m]|. 
\end{equation*}

Formally, we define a map $f$ from the uniform distribution on all orderings to the uniform distribution 
on all orderings such that $e\in\pi (\beta m, m] $: if $e\in \pi (\beta m, m] $ then $f(\pi)=\pi$ and otherwise 
$f(\pi)$ is the permutation obtained from $\pi$ by removing $e$ and re-inserting it at a position picked uniformly 
at random in $ (\beta m, m] $; in the latter case, if this causes an edge $e'=a'b'$, previously arriving at time 
$\lfloor \beta m\rfloor +1$, to now arrive at time $\lfloor \beta m\rfloor$ and to be added to $M_1$, we define 
$A''=A'\setminus \{ M_0(b') \}$; in all other cases we define $A''=A'$. 
Thus, if in $\pi$ we have $e \in \OPT(A', \overline{B(M_0)})$, then in $f(\pi)$ we have 
$e \in \OPT(A'', \overline{B(M_0)})$. Since the distribution of $f(\pi)$ is uniform conditioned on 
$e\in \pi(\beta m, m] $: 
\begin{eqnarray*}
\frac{ \Pr [ e\in \OPT(A'', \overline{B(M_0)}) \hbox{ and } e\in \pi(\beta  m,m] ] }{ \Pr [ e\in \pi(\beta m, m] ] }
\geq
 \Pr [e\in \OPT(A', \overline{B(M_0)})].
\end{eqnarray*}
%
Using $\Pr [ e\in \pi(\beta m, m] ] = 1-\beta$ and summing over $e$:
\begin{eqnarray*}
\Exp_{\pi}  |\OPT(A'', \overline{B(M_0)})\cap \pi(\beta m,  m]|\geq (1-\beta ) \Exp_{\pi}  |\OPT(A', \overline{B(M_0)})|. 
\end{eqnarray*}

Since $A'$ and $A''$ differ by at most one vertex, 
\begin{eqnarray*}
|\OPT(A'', \overline{B(M_0)})|\geq |\OPT(A', \overline{B(M_0)})|-1, 
\end{eqnarray*}
and the Lemma follows.  \qed
\end{proof}

We now present the proof of the main theorem, Theorem~\ref{theorem:one-pass-bip}.
\begin{theorem}
\label{theorem:one-pass-bip} Algorithm~\ref{l_algo1} is a deterministic one-pass semi-streaming
algorithm for \mbmshort{} with expected approximation ratio $\frac{1}{2} + 0.005$ 
against (uniform) random order for any graph,
and can be implemented with $\Order(1)$ update time. 
\end{theorem}

\begin{proof}
Assume that $\Exp_{\pi}|M_G| \leq (\frac{1}{2} + \epsilon) |M^*|$. 
By construction, every $e \in M_2$ completes a $3-$augmenting path, hence 
$ |M| \ge  |M_0| + |M_2|.$
In Lemma~\ref{lemma:exp} we show that $\Exp_{\pi} |M_0| \ge |M^*|(\frac{1}{2}-(\frac{1}{\alpha }-2) \epsilon)$. 
By Lemmas~\ref{l_left_wings} and~\ref{l_matching_size_left_wings}, $|M_2|$ can be related to $|M_1|$:
\begin{eqnarray*}
\Exp_{\pi} |M_2|  \geq \frac{1}{2} (1 - \beta) \Exp_{\pi} |\OPT(A', \overline{B(M_0)})| -\frac{1}{2}
 \ge \frac{1}{2} (1 -\beta ) (  \Exp_{\pi} |M_1| - 4 \epsilon |M^*|  ) -\frac{1}{2}. 
\end{eqnarray*}
By Lemmas~\ref{l_exp_right_wings} and~\ref{l_number_of_right_wings}, $|M_1|$ can be related to $|M^*|$:
\begin{eqnarray*}
 \Exp_{\pi} |M_1| & \geq & \tfrac{1}{2} (\beta-\alpha ) \Exp_{\pi} |\OPT(\overline{A(M_0)}, B(M_0)| -O(1) \\
& \ge & \tfrac{1}{2} (\beta-\alpha )(  |M^*|(\frac{1}{2}-(\frac{1}{\alpha}+2) \epsilon)  )-O(1).
\end{eqnarray*}

%
Combining,\\
\begin{eqnarray*}
\Exp_{\pi} |M| \ge \hspace{7cm} \quad \\
|M^*| ( \tfrac{1}{2} - (\tfrac{1}{\alpha } -2 ) \epsilon + \tfrac{1}{2}(1 -\beta ) ( \tfrac{1}{2} (\beta-\alpha)(\tfrac{1}{2} - (\tfrac{1}{\alpha} + 2) \epsilon) - 4 \epsilon ) ) -O(1).
\end{eqnarray*}
The expected value of the output of the Algorithm is at least $\min_\epsilon \max \{ (\frac{1}{2} + \epsilon) |M^*|, \Exp_{\pi} |M| \}$.
We set the right hand side of the above Equation equal to $(\frac{1}{2} + \epsilon) |M^*|$.
By a numerical search we optimize parameters $\alpha, \beta$. 
Setting $\alpha = 0.4312$ and $\beta = 0.7595$, we obtain $\epsilon \approx 0.005$ 
which proves the Theorem.  \qed
\end{proof}

\subsection{Extension to General Graphs} \label{section:one-pass-general}
In this section, we show how the one-pass algorithm of Section~\ref{section:one-pass-bipartite} can be adapted
to general graphs $G = (V, E)$.

\subsubsection{Algorithm}
Algorithm~\ref{l_algo4} follows the same line as Algorithm~\ref{l_algo1} for the bipartite case.
While in the bipartite case, edges from $M_1$ extend $M_0$ on only one bipartition, and those edges do not
interfere with edges from $M_2$, this structure is no longer given in the general setting. Here, $M_1$
is a $\Greedy$ matching between the matched vertices in $M_0$ and all free vertices. This may already
produce some $3$-augmenting paths, however, it may also happen that by taking a {\em bad} edge into $M_1$, 
this rules out any possibility of finishing the $3$-augmenting paths containing these edges. We call the edge 
of $M_0$ \textit{blocked} if it can not be completed to a $3$-augmenting path, see Definition~\ref{def_blocked_edge}.

\begin{figure}[ht]
\centerline{\scalebox{0.9}{
 \begin{tikzpicture}[node distance = 1.5 cm]
     \GraphInit[vstyle=normal]
      \Vertex[x=0, y=0]{a}
      \Vertex[x=2.5, y=0]{b}
      \Vertex[x=5, y=0]{c}
      \Vertex[x=7.5, y=0]{d}
      \SetUpEdge[style={ultra thick}, color=red]
      \Edge[label=$\in M_0$](b)(c)          
      \SetUpEdge[style={dashed}, color=black]
      \Edge(a)(b)          
      \Edge(c)(d)          
      \SetUpEdge[style={solid,bend left=45},color=black]
      \Edge(b)(d)
  \end{tikzpicture}
}}
\caption[Example of a Blocked Edge]{If edge $bd$ is taken into $M_1$ and edge $ac \notin E$, this may block the $3$-augmenting path $ab,bc,cd$. In that case we call $bc$ blocked.}
\end{figure}
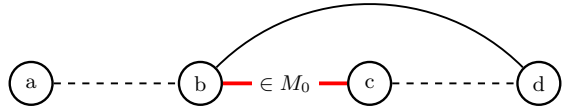

We show in Lemma~\ref{l_few_blocked_edges} that the probability that an edge of $M_0$ will become blocked is at most $1/2$. 
This guarantees that we can finalize many $3$-augmenting paths by the $\Greedy$ matching $M_2$.

\begin{algorithm}[ht] \small
\caption{One-pass Matching on Random Order for General Graphs \label{l_algo4}}
\begin{algorithmic}[1]
\STATE $\alpha \gets 0.413, \beta \gets 0.708$
\STATE $M_G \gets \Greedy(\pi)$
\STATE $M_0 \gets \Greedy(\pi[1,\alpha m])$,  matching obtained by $\Greedy$ on the first $\lfloor \alpha m \rfloor$ edges
\STATE $F_1 \gets $ complete bipartite graph between $V(M_0)$ and $\overline{V(M_0)}$
\STATE $M_1 \gets \Greedy(F_1 \cap\pi(\alpha m, \beta m])$,  matching obtained by Greedy on edges $\lfloor \alpha m\rfloor +1$ through $\beta m$ that intersect $F_1$
\STATE $Aug \gets $ length $3$ paths in $M_0 \oplus M_1$
\STATE $V_1 \gets \{ u \in V \setminus V(Aug) \, | \, \exists v \in V(M_1): uv \in M_0 \}$
\STATE $V_2 \gets \overline{V(M_0)} \setminus V(Aug)$
\STATE   $F_2 \gets $  maximal bipartite graph between $V_1$ and $V_2$ such that $\nexists\,m_0 
\in M_0 \setminus Aug, m_1 \in M_1 \setminus Aug, f_2 \in F_2$ st. they form a triangle
\STATE $M_2 \gets \Greedy(F_2 \cap \pi(\beta m, m])$,  matching obtained by $\Greedy$ on edges $\lfloor \beta m\rfloor +1$ through $m$ that intersect $F_2$
\STATE  $M\gets $  matching obtained from $M_0$ augmented by $M_1 \cup M_2$
\RETURN larger of the two matchings  $M_G$ and $M$
\end{algorithmic}
\end{algorithm}

$Aug$ is a set of length $3$ paths. $|Aug|$ denotes the number of length $3$ paths in $Aug$.
For some vertex $a \in V$ (resp. some edge $e \in E$), we write $a \in Aug$ (resp. $e \in Aug$) 
if $a$ (resp. $e$) is part of some length $3$ path.

\subsubsection{Analysis}
We bound the size of a maximum matching between $V(M_0)$ and $\overline{V(M_0)}$.

\begin{lemma}
 \label{l_f1_size}
Assume that $\Exp_\pi |M_G|\leq (\frac{1}{2} + \epsilon) |M^*|$. Then:
 $$\Exp | \MM(V(M_0), \overline{V(M_0))} | \ge |M^*| (1 - 2(\frac{1}{\alpha}+2) \epsilon) .$$
\end{lemma}

\begin{proof}
 The size of a maximum matching between $\overline{V(M_0)}$ and $ V(M_0)$ 
is at least twice the number of augmenting paths of length 3 in $M_0\oplus M^*$. 
By Lemma~\ref{lemma:augmentable_edges}, in expectation, the number of augmenting 
paths of length $3$ in $ M_G  \oplus M^*$ is at least 
$ (\frac{1}{2} - 3 \epsilon)|M^*|$.
All of those are augmenting paths of length $3$ in $M_0\oplus M^*$, 
except for at most $|M_G| - |M_0|$. Hence, in expectation, $M_0$ contains
$(\frac{1}{2} - 3 \epsilon)|M^*| - (\Exp |M_G| - \Exp |M_0|)$ edges that are
$3$-augmentable. Lemma~\ref{lemma:exp} applied to $M_0$ concludes the proof.  \qed
\end{proof}

\begin{lemma}
 \label{l_m1_size}
Assume that $\Exp_\pi |M_G|\leq (\frac{1}{2} + \epsilon) |M^*|$. Then:
 $$\Exp |M_1| \ge \frac{1}{2}(\beta - \alpha) (\Exp | \MM(V(M_0), \overline{V(M_0)}) | - \frac{1}{1-\alpha}).$$
\end{lemma}

\begin{proof}
 The proof is identical to the proof of Lemma~\ref{l_exp_right_wings}. \qed
\end{proof}

\begin{definition}[Blocked edge]
\label{def_blocked_edge}
 Let $e=uv \in M_0$ such that $e$ is $3$-augmentable by edges $o_1=uu', o_2=vv' \in M^*$.
 We call $e$ \textit{blocked}, if:
\begin{enumerate}
 \item either $uv' \in E$ or $u'v \in E$ (not both of them), and
 \item if $uv' \in E$ then $uv' \in M_1$, otherwise $u'v \in M_1$.
\end{enumerate}
\end{definition}

\begin{lemma}
\label{l_few_blocked_edges}
 $$\Pr [e \mbox{ blocked } \, | \, e \in M_0 ] \le \frac{1}{2}.$$
\end{lemma}

\begin{proof}
 W.l.o.g. let $uv' \in E$ and $u'v \notin E$. 
\begin{eqnarray*}
 \Pr [e \mbox{ blocked } \, | \, e \in M_0 ] & = & \Pr [e \notin Aug \mbox{ and } uv' \in M_1 \, | \, e \in M_0 ] \\
 & \le & \Pr [uv' \in M_1 \, | \, e \in M_0 \setminus Aug].
\end{eqnarray*}
Since $\Pr [uv' \in M_1 \, | \, e \in M_0 \setminus Aug] = \Pr [vv' \in M_1 \, | \, e \in M_0 \setminus Aug]$, and 
since the events $(uv' \in M_1 \, | \, e \in M_0 \setminus Aug)$ and $(vv' \in M_1 \, | \, e \in M_0 \setminus Aug)$ 
exclude each other, the result follows. \qed
\end{proof}

\begin{lemma}
 \label{l_f2_size}
 \begin{equation*}
\Exp |\MM(F_2)| \ge \max \{ \frac{1}{2}  (\Exp |M_1| - 4 |Aug| - 4\epsilon |M^*|), 0 \} .  
 \end{equation*}
\end{lemma}

\begin{proof}
The size of a maximum matching in $F_2$ is at least the number of length $2$ paths in $M_0 \oplus M_1$ 
that can be completed to a $3$-augmenting path.
Denote by $k_2$ the number of length two paths in $M_0 \oplus M_1$. 
Then, $|M_1| = 2|Aug| + k_2$. A length $3$ path may block at most $2$ other length 2 
paths from being completed. 

By Lemma~\ref{lemma:augmentable_edges}, the number of edges of $|M_G|$ that are not $3$-augmentable is
in expectation at most $(\frac{1}{2} + \epsilon) |M^*| - (\frac{1}{2} - 3\epsilon)|M^*| = 4 \epsilon |M^*|$.
Since $M_0$ is a subset of $M_G$, it follows that at most $4 \epsilon |M^*|$ edges from $M_0$ are 
not $3$-augmentable. Hence, the number of $M_0$ edges for which a length two path was 
found and which is $3$-augmentable is at least $(k_2 - 2 |Aug| - 4\epsilon |M^*|)$. 
In expectation, by Lemma~\ref{l_few_blocked_edges}, at most half of these edges are blocked. The Lemma follows. \qed
\end{proof}

\begin{lemma}
 \label{l_m2_size}
 \begin{equation*}
\Exp |M_2| \ge \frac{1}{2} \left( (1-\beta) \Exp |\MM(F_2)| - 1 \right).  
 \end{equation*}
\end{lemma}
\begin{proof}
 This proof is identical to the proof of Lemma~\ref{l_left_wings}. \qed
\end{proof}

We now present the proof of the main theorem, Theorem~\ref{theorem:one-pass-gen}.
\begin{theorem}
 \label{theorem:one-pass-gen}
  Algorithm~\ref{l_algo4} is a deterministic one-pass semi-streaming algorithm 
for \mm~with approximation ratio $\frac{1}{2} + 0.00363$ in expectation over (uniform)
random order for any graph, and can be implemented with $\Order(1)$ update
time.
\end{theorem}

\begin{proof}
 The expected matching size is
\begin{equation}
 \Exp |M| \ge \Exp |M_0| + |Aug| +  \frac{1}{2} \Exp |M_2|, \label{eqn:2011}
\end{equation}
since, by construction, at least half of the edges of $M_2$ can be used to complete a $3$-augmenting path.
Firstly, we bound $|M_2|$ by Lemma~\ref{l_m2_size} and Lemma~\ref{l_f2_size} and we obtain
\begin{equation}
 \Exp |M_2| \ge \max \{ 0, (1-\beta)(\frac{1}{4}  \Exp |M_1| - |Aug| - \epsilon |M^*|) - \Order(1) \}. \label{eqn:5921}
\end{equation}
By Lemma~\ref{l_m1_size} and Lemma~\ref{l_f1_size}, we bound the size of $M_1$ and we obtain
\begin{equation}
 \Exp |M_1|  \ge \frac{1}{2} |M^*| (\beta - \alpha) (1 - 2(\frac{1}{\alpha}+2) \epsilon) - \Order(1). \label{eqn:4441}
\end{equation}
Using Inequality~\ref{eqn:4441} in Inequality~\ref{eqn:5921}, we obtain
\begin{equation}
 \Exp |M_2| \ge \max\{0, (1-\beta)(\frac{1}{8} |M^*| \left( (\beta - \alpha) (1 - 2(\frac{1}{\alpha}+2) \epsilon) -\epsilon \right) - |Aug|) - \Order(1) \}. \label{eqn:3911}
\end{equation}
Furthermore, in Lemma~\ref{lemma:exp} we show that $\Exp_{\pi} |M_0| \ge |M^*|(\frac{1}{2}-(\frac{1}{\alpha }-2) \epsilon)$. 
We use this and Inequality~\ref{eqn:3911} in Inequality~\ref{eqn:2011} and we obtain an Inequality for $\Exp |M|$ that depends
on $\alpha, \beta, |Aug|$ and $\epsilon$. It is easy to see that this Inequality is minimized if $|Aug| = 0$. 

The expected value of the output of the Algorithm is at least $\min_\epsilon \max \{ (\frac{1}{2} + \epsilon) |M^*|, \Exp_{\pi} |M| \}$.
By a numerical search we optimize parameters $\alpha, \beta$. Setting $\alpha = 0.413, \beta = 0.708$, we obtain $\epsilon \approx 0.00363$.
which proves the Theorem.  \qed

\end{proof}

\section{Randomized Two-pass Algorithm on any Order} \label{section:matching-two-pass-randomized}

We present now a randomized two-pass semi-streaming algorithm for \mbm{}
with approximation ratio strictly greater than $\frac{1}{2}$. This algorithm simulates the three passes of the 
$3$-pass algorithm of Section~\ref{section:matching-3-pass}
in two passes. We require a new property of the $\Greedy$ algorithm that may be of independent interest.
In Subsection~\ref{section:matching-subset}, we discuss this new property. Then, we present in Subsection~\ref{subsection:matching-two-passes-rand-bip}
our two-pass randomized algorithm for bipartite graphs.

\subsection{Matching Many Vertices of a Random Vertex Subset} \label{section:matching-subset}

Consider a bipartite graph $G=(A, B, E)$. For a fixed parameter $0<p\leq 1$, Algorithm~\ref{algo:matching-random-subset}
generates an independent random sample of vertices $A'\subseteq A$ such that $\Pr [a \in A'] = p$, for all $a\in A$,
and runs then the Greedy algorithm on the subgraph $G|_{A' \times B}$.

\begin{algorithm} 
\caption{Matching a Random Subset of Vertices (Bipartite Graphs) \label{algo:matching-random-subset}}
\begin{algorithmic}[1]
\STATE Take independent random sample $A'\subseteq A$ st. $\Pr [a \in A'] = p$, for all $a\in A$
\STATE Let $F$ be the complete bipartite graph between $A'$ and $B$
 \RETURN $M' = \Greedy(F\cap \pi)$ 
\end{algorithmic}
\end{algorithm}

We prove in Theorem~\ref{theorem:subset-match} that
the greedy algorithm restricted to the edges with an endpoint 
in $A'$ will output a matching of expected approximation ratio $p/(1+p)$, 
compared to a maximum matching $\OPT(G)$ over the full graph $G$.
Since, in expectation, the size of $A'$ is $p|A|$, one can roughly say
that a fraction of $1/(1+p)$ of vertices in $|A'|$ has been matched.

The proof of Theorem~\ref{theorem:subset-match} will use Wald's equation for super-martingales, see \cite{mu05},
Wald's Equation, p.300, section 12.3.\footnote{The theorem cited in the book is actually 
weaker than the one we need, but our statement follows from the proof of that 
Theorem.} 

\begin{lemma}[Wald's equation]  \label{lemma:wald} 
Consider a process described by a sequence of random states $(S_i)_{i\geq 0}$ and let $D$ be a random stopping time for the process, 
such that $\Exp D <\infty$. Let $(\Phi(S_i))_{i\geq 0}$ be a sequence of random variables for which there exist $c, \mu$ such that
\begin{enumerate}
 \item $\Phi(S_0)=0$;
 \item $\Phi (S_{i+1})-\Phi(S_i) <c$ for all $i<D$; and 
 \item $\Exp [ \Phi (S_{i+1})-\Phi (S_i) \, | \, S_i ]\leq \mu$ for all $i<D$.
\end{enumerate}
Then:
$$\Exp \Phi(S_D) \leq \mu \Exp D.$$
\end{lemma}

\begin{theorem}
\label{theorem:subset-match}
Let $0 < p \le 1$, let $G = (A, B, E)$ be a bipartite graph. Let $A'$ 
be an independent random sample $A' \subset A$ such that $\Pr [a \in A'] = p$, for all $a\in A$.
Let $F$ be the complete bipartite graph between $A'$ and $B$
Then for any input stream $\pi \in \Pi(G)$:
\begin{equation*}
\displaystyle \Exp_{A'} |\Greedy(F \cap \pi)| \ge \frac{p}{1+p} |\OPT(G)|. 
\end{equation*}
\end{theorem}

\begin{proof}
Let $M' = \Greedy(F \cap \pi)$.  For $i\leq |M'|$, 
denote by $M'_i$ the first $i$ edges of $M'$, in the order in which they were added to $M'$ during the execution of Greedy. 

Let $M^*$ be a fixed maximum matching in $G$ and let $M_F$ denote the edges of $M^*$ that are in $F$.   
Let $A'' = A(M_F)$ denote the vertices of $A'$ matched by $M_F$. Consider a vertex $a \in A''$ and its match $b$ in matching $M_F$. We say that $a$ is \textit{live} with
respect to $M'_i$ if both $a$ and $b$ are  unmatched in 
$M'_i$. A vertex that is not live is \textit{dead}. 
Furthermore, we say that an 
edge of $M'_{i+1}\setminus M'_i$ \textit{kills} a  vertex $a$ if $a$ is
live with respect to $M'_i$ and dead with respect to $M'_{i+1}$.  

We use Lemma~\ref{lemma:wald}. Here, by ``time``, we mean the number of edges in $M'$, so between time $i-1$ and time $i$, during the execution of Greedy,
several edges arrive and all are rejected except the last one which is added to $M'$. 
We use a potential function $\phi(i)$ 
which we define as the number of dead vertices  with respect to $M'_i$. We define the stopping
time $D$ as the first time when the event  $\phi (i)=|A''|$ holds.

We only need to check that the three assumptions of the Stopping Lemma hold. First, initially  all nodes of $A''$ are live, so $\phi(0)=0$.
Second, the potential function $\phi$ is  non-decreasing and uniformly bounded: since adding an edge to $M'$ can kill at most two vertices of $A''$, 
we always have 
$\Delta \phi (i) := \phi(i+1)-\phi(i)\leq 2$. 
Third, let $S_i$ denote the state of the process at time $i$,
namely the information about the entire sequence of edge arrivals up to that time, 
hence, in particular, the set of $i$ edges currently in $M'$.
Observe that, here, $G$ and $M^*$ are fixed.
Then $D$ is indeed a stopping time, since the event $D\geq i+1$ can be inferred from the knowledge of $S_i$.
We now claim that:
\begin{equation}
\label{l_decrease_of_phi}
\Exp (\Delta \phi(i) ~|~ S_i)\leq 1+p. 
\end{equation}
Indeed, since $\Delta\phi(i)$ only takes on values 0, 1 or 2, we can write that 
\begin{equation*}
\Exp (\Delta \phi(i) ~|~ S_i) \le  1+\Pr[\Delta\phi (i)=2 ~|~ S_i].  
\end{equation*}

To bound the latter probability, let $e=ab$ denote the edge of $M'_{i+1} \setminus M'_i$ and let $t$ be such that $e = \pi[t]$. 
In order for $e$ to change $\phi$ by $2$, it must be that $b$ is matched in $M^*$ to a node $a'$ that is also in $A''$.
Furthermore, it is required that $a'$ was unmatched before edge $e$ arrived. 
Since $a'$ was unmatched up to arrival $t$, no edge $a'b'$ had been seen among the first $t$ edges of stream $\pi$, such that $b'$ was free at arrival time (of $a'b'$).
Thus
\begin{eqnarray*}
& &  \Pr [ \Delta\phi (i)=2 ~|~ S_i]\leq \\
& & \hspace{1cm} \Pr [ a'\in A'\mbox{ and $\nexists a'b' \in \pi[1,t]$ st. $b'$ was free when $a'b'$ arrived} ~|~ S_i].
\end{eqnarray*}
Now, given that no edge $f=a'b'$ arrived before $t$ such that $b'$ was free when $a'b'$ arrived,
the outcome of the random coin determining whether $a' \in A'$ was never looked at, 
and could have been postponed until $t$. Thus
\begin{eqnarray*}
 & & \Pr[a' \in A' \, | \, (\mbox{$\nexists a'b' \in \pi[1,t]$ such that $b'$ was free when $a'b'$ arrived}, S_i ) ] = \\
 & & \hspace{1cm}  \Pr [a'\in A']=p,
\end{eqnarray*}
implying Inequality~\ref{l_decrease_of_phi}.
Applying  Wald's Stopping Lemma, we obtain 
\begin{eqnarray*}
\Exp \phi(D) \leq (1+p) \Exp D. 
\end{eqnarray*}
Finally, since $ \Exp \phi(D) = \Exp |A''| = p\cdot |\OPT(G)|$ and  
$\mbox{$D\leq  |\Greedy(F\cap \pi)|$}$, and the Theorem follows.  \qed
\end{proof}

\subsection{A Randomized Two-pass Algorithm for Bipartite Graphs} \label{subsection:matching-two-passes-rand-bip}
Based on Theorem~\ref{theorem:subset-match}, we design our randomized two-pass algorithm for bipartite graphs $G=(A, B, E)$.
Assume that $\Greedy(\pi)$ returns a matching that is close to 
a $\frac{1}{2}$-approximation. In order to apply Theorem~\ref{theorem:subset-match}, 
we pick an independent random sample $A' \subseteq A$ such that $\Pr[a \in A'] = p$ for all $a$. 
In a first pass, our algorithm computes a $\Greedy$ matching $M_0$ of $G$, and a $\Greedy$ matching $M'$
between vertices of $A'$ and $B$. $M'$ then contains some edges that form parts of $3$-augmenting
paths for $M_0$: see Figure~\ref{figure:two-pass-rand-bip-2} and Figure~\ref{figure:two-pass-rand-bip} 
for an illustration. Let $M_1 \subset M'$ be the set of those edges. 
It remains to complete these length $2$ paths $M_0 \cup M_1$ in a second pass by a further $\Greedy$ matching $M_2$.
In the prove of Theorem~\ref{l_theorem2}, we show that if $\Greedy(\pi)$ is close to a $\frac{1}{2}$-approximation, then we 
find many $3$-augmenting paths.

\begin{algorithm}
\caption{Two-pass Randomized Bipartite Matching Algorithm \label{algorithm:two-pass-rand-bip}}
\begin{algorithmic}[1]
\STATE Let $p\gets \sqrt{2} - 1$.   
\STATE Take an independent random sample $A' \subseteq A$ st. $\Pr [a \in A'] = p$, for all $a\in A$
\STATE  Let $F_1$ be the set of edges with one endpoint in $A'$. 
 \STATE \textbf{First pass: } $M_0 \gets \Greedy(\pi)$ and $M' \gets \Greedy(F_1\cap \pi)$
 \STATE $M_1\gets \{ e \in M' \, | \, e \mbox{ goes between } B(M_0)  \mbox{ and } \overline{A(M_0)} \}$
 \STATE $A_2 \gets \{ a \in A(M_0): \exists b,c: ab \in M_0 \mbox{ and } bc \in M_1 \}$. 
 \STATE  Let $F_2\gets \{ da:  d\in \overline{B(M_0)}\mbox{ and } a \in A(M_0)\mbox{ and }\exists b,c: ab \in M_0 \mbox{ and } bc \in M_1  \}$.
 \STATE \textbf{Second pass: } $M_2 \gets \Greedy(F_2\cap \pi)$
 \STATE Augment $M_0$ by edges in $M_1$ and $M_2$ and store it in $M$
 \RETURN the resulting matching $M$
\end{algorithmic}
\end{algorithm}

\begin{figure}[h!]
\centerline{\scalebox{0.9}{
\begin{tikzpicture}[node distance = 1.5 cm]
    \GraphInit[vstyle=normal]
\draw(-2.0, 8.75) node{$3$-augmentable edges};
\draw(-2.0, 5.5) node{other edges};
\draw(2.25, 11) node{\textbf{$M_0$}};
\draw(4.35, 11) node{\textbf{$M_1 \subseteq M'$}};
\draw(0, 10.5) node{\textbf{B}};
\draw(1.5, 10.5) node{\textbf{A}};
\draw(3, 10.5) node{\textbf{B}};
\draw(4.5, 10.5) node{\textbf{A}};
\draw(2.25, 7) node{$\vdots$};
\draw(2.25, 5.5) node{$\vdots$};
\draw(8.3, 9) node{vertex $\in A'$};
\draw(8.3, 8) node{edge $\in M' \setminus M_1$};
\draw(8.3, 8.5) node{edge $\in M_1$};
\draw(8.3, 7.5) node{edge $\in M^*$};
\tikzset{VertexStyle/.style = {shape = circle, color = black, fill = white, minimum size = 0.1cm, draw}}
\Vertex[x=0, y=10, L=$$]{d1} \Vertex[x=1.5, y=10, L=$$]{a1} \Vertex[x=3, y=10, L=$$]{b1} 
\Vertex[x=0, y=9.5, L=$$]{d2} 				 \Vertex[x=3, y=9.5, L=$$]{b2} \Vertex[x=4.5, y=9.5, L=$$]{c2}
\Vertex[x=0, y=9, L=$$]{d3}  \Vertex[x=3, y=9, L=$$]{b3} \Vertex[x=4.5, y=9, L=$$]{c3}                
\Vertex[x=0, y=8.5, L=$$]{d4} \Vertex[x=1.5, y=8.5, L=$$]{a4} \Vertex[x=3, y=8.5, L=$$]{b4}                 
\Vertex[x=0, y=8, L=$$]{d5} \Vertex[x=1.5, y=8, L=$$]{a5} \Vertex[x=3, y=8, L=$$]{b5} \Vertex[x=4.5, y=8, L=$$]{c5}                
\Vertex[x=0, y=7.5, L=$$]{d6} 				 \Vertex[x=3, y=7.5, L=$$]{b6} \Vertex[x=4.5, y=7.5, L=$$]{c6}
\Vertex[x=0, y=6.5, L=$$]{d7} \Vertex[x=1.5, y=6.5, L=$$]{a7} \Vertex[x=3, y=6.5, L=$$]{b7} \Vertex[x=4.5, y=6.5, L=$$]{c7}
\Vertex[x=1.5, y=6, L=$$]{a8} \Vertex[x=3, y=6, L=$$]{b8} 
\Vertex[x=1.5, y=5, L=$$]{a9} \Vertex[x=3, y=5, L=$$]{b9} 
\Vertex[x=0, y=6, L=$$]{d8}
\tikzset{VertexStyle/.style = {shape = circle, color = black, fill = gray, minimum size = 0.1cm, draw}}
\Vertex[x=1.5, y=9.5, L=$$]{a2}
\Vertex[x=1.5, y=7.5, L=$$]{a6}
\Vertex[x=4.5, y=10, L=$$]{c1}
\Vertex[x=4.5, y=8.5, L=$$]{c4}
\Vertex[x=5.95, y=9, L=$$]{e1}
\Vertex[x=5.3, y=8.5, L=$$]{e2}
\Vertex[x=1.5, y=9, L=$$]{a3}
\tikzset{VertexStyle/.style = {shape = circle, color = white, fill = white, minimum size = 0.1cm, draw}}
\Vertex[x=6.8, y=8.5, L=$$]{f2}
\Vertex[x=6.8, y=7.5, L=$$]{f3}
\Vertex[x=5.3, y=7.5, L=$$]{e3}
\Vertex[x=6.8, y=8, L=$$]{f4}
\Vertex[x=5.3, y=8, L=$$]{e4}
\SetUpEdge[style={solid}, color=black]
\Edge(a1)(b1)
\Edge(a2)(b2)
\Edge(a3)(b3)
\Edge(a4)(b4)
\Edge(a5)(b5)
\Edge(a6)(b6)
\Edge(a7)(b7)
\Edge(a8)(b8)
\Edge(a9)(b9)
\SetUpEdge[style={dashed}, color=black]
\Edge(d1)(a1) \Edge(b1)(c1)
\Edge(d2)(a2) \Edge(b2)(c2)
\Edge(d3)(a3) \Edge(b3)(c3)
\Edge(d4)(a4) \Edge(b4)(c4)
\Edge(d5)(a5) \Edge(b5)(c5)
\Edge(d6)(a6) \Edge(b6)(c6)
\Edge(d7)(a7) \Edge(b7)(c7)
\Edge(e3)(f3)
\Edge(a8)(d8)
\SetUpEdge[style={thick}, color=black]
\Edge(c1)(b2)
\Edge(c4)(b6)
\Edge(e2)(f2)
\SetUpEdge[style={dotted}, color=black]
\Edge(a2)(b3)
\Edge(a6)(b8)
\Edge(a3)(d5)
\Edge(e4)(f4)
 \end{tikzpicture}}}
\caption[First Pass of the Randomized Two-pass Bipartite Matching Algorithm]{Illustration of the first pass of Algorithm~\ref{algorithm:two-pass-rand-bip}. By Theorem~\ref{theorem:subset-match}, nearly all vertices of $A'$ are matched in $M'$, in particular those that are not matched in $M_0$. \label{figure:two-pass-rand-bip-2} } \end{figure}
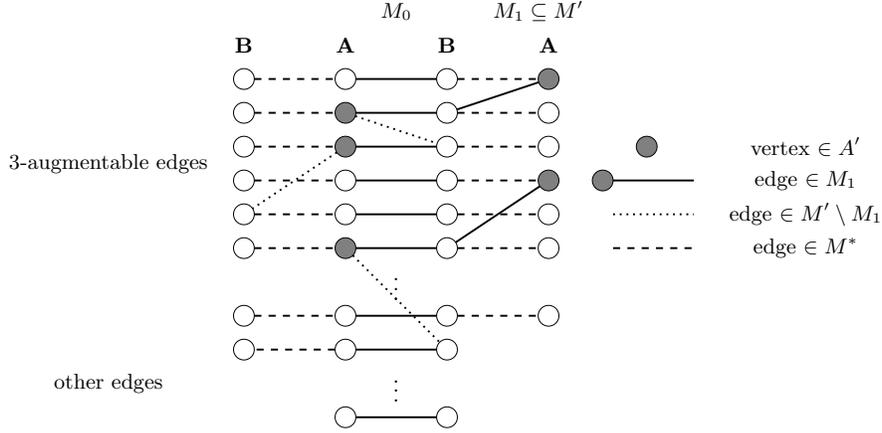

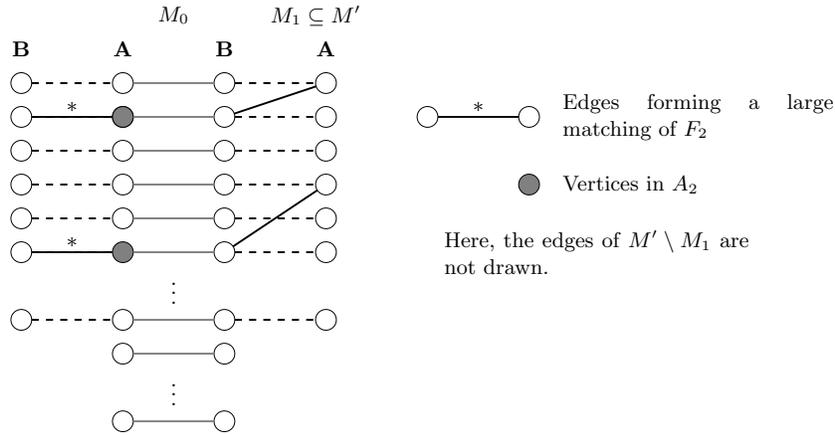
\begin{figure}[ht]
\centerline{\scalebox{0.9}{
\begin{tikzpicture}[node distance = 1.5 cm]
    \GraphInit[vstyle=normal]
%
\draw(2.25, 11) node{\textbf{$M_0$}};
\draw(4.35, 11) node{\textbf{$M_1 \subseteq M'$}};
\draw(0, 10.5) node{\textbf{B}};
\draw(1.5, 10.5) node{\textbf{A}};
\draw(3, 10.5) node{\textbf{B}};
\draw(4.5, 10.5) node{\textbf{A}};
\draw(2.25, 7) node{$\vdots$};
\draw(2.25, 5.5) node{$\vdots$};
\draw(6.75, 9.6) node{*};
\draw(0.75, 9.6) node{*};
\draw(0.75, 7.6) node{*};
\draw(8.5, 7.5) node{\parbox{4.5cm}{Here, the edges of $M' \setminus M_1$ are not drawn.}};
\draw(10, 9.5) node{\parbox{4cm}{Edges forming a large matching of $F_2$}};
\draw(10, 8.5) node{\parbox{4cm}{Vertices in $A_2$}};
\tikzset{VertexStyle/.style = {shape = circle, color = black, fill = white, minimum size = 0.1cm, draw}}
\Vertex[x=0, y=10, L=$$]{d1} \Vertex[x=1.5, y=10, L=$$]{a1} \Vertex[x=3, y=10, L=$$]{b1} 
\Vertex[x=0, y=9.5, L=$$]{d2} 				 \Vertex[x=3, y=9.5, L=$$]{b2} \Vertex[x=4.5, y=9.5, L=$$]{c2}
\Vertex[x=0, y=9, L=$$]{d3}  \Vertex[x=3, y=9, L=$$]{b3} \Vertex[x=4.5, y=9, L=$$]{c3}                
\Vertex[x=0, y=8.5, L=$$]{d4} \Vertex[x=1.5, y=8.5, L=$$]{a4} \Vertex[x=3, y=8.5, L=$$]{b4}                 
\Vertex[x=0, y=8, L=$$]{d5} \Vertex[x=1.5, y=8, L=$$]{a5} \Vertex[x=3, y=8, L=$$]{b5} \Vertex[x=4.5, y=8, L=$$]{c5}                
\Vertex[x=0, y=7.5, L=$$]{d6} 				 \Vertex[x=3, y=7.5, L=$$]{b6} \Vertex[x=4.5, y=7.5, L=$$]{c6}
\Vertex[x=0, y=6.5, L=$$]{d7} \Vertex[x=1.5, y=6.5, L=$$]{a7} \Vertex[x=3, y=6.5, L=$$]{b7} \Vertex[x=4.5, y=6.5, L=$$]{c7}
\Vertex[x=1.5, y=6, L=$$]{a8} \Vertex[x=3, y=6, L=$$]{b8} 
\Vertex[x=1.5, y=5, L=$$]{a9} \Vertex[x=3, y=5, L=$$]{b9} 
\Vertex[x=6, y=9.5, L=$$]{e2}
\tikzset{VertexStyle/.style = {shape = circle, color = black, fill = white, minimum size = 0.1cm, draw}}
\Vertex[x=1.5, y=9.5, L=$$]{a2}
\Vertex[x=4.5, y=10, L=$$]{c1}
\Vertex[x=4.5, y=8.5, L=$$]{c4}
\Vertex[x=7.5, y=9.5, L=$$]{f2}
\Vertex[x=1.5, y=9, L=$$]{a3}
%
\tikzset{VertexStyle/.style = {shape = circle, color = black, fill = gray, minimum size = 0.1cm, draw}}
\Vertex[x=1.5, y=9.5, L=$$]{a2}
\Vertex[x=1.5, y=7.5, L=$$]{a6}
\Vertex[x=7.5, y=8.5, L=$$]{f3}
\SetUpEdge[style={solid}, color=gray]
\Edge(a1)(b1)
\Edge(a2)(b2)
\Edge(a3)(b3)
\Edge(a4)(b4)
\Edge(a5)(b5)
\Edge(a6)(b6)
\Edge(a7)(b7)
\Edge(a8)(b8)
\Edge(a9)(b9)
\SetUpEdge[style={dashed}, color=black]
\Edge(d1)(a1) \Edge(b1)(c1)
\Edge(d2)(a2) \Edge(b2)(c2)
\Edge(d3)(a3) \Edge(b3)(c3)
\Edge(d4)(a4) \Edge(b4)(c4)
\Edge(d5)(a5) \Edge(b5)(c5)
\Edge(d6)(a6) \Edge(b6)(c6)
\Edge(d7)(a7) \Edge(b7)(c7)
\SetUpEdge[style={thick}, color=black]
\Edge(c1)(b2)
\Edge(c4)(b6)
\Edge(a2)(d2)
\Edge(a6)(d6)
\Edge(e2)(f2)
 \end{tikzpicture} }}
\caption[Second Pass of the Randomized Two-pass Bipartite Matching Algorithm]{Analysis of the second pass of Algorithm~\ref{algorithm:two-pass-rand-bip}. Here, we see that $M_0\oplus
M_1$ has two paths of length 2, and that both of those paths can be
extended into 3-augmenting paths using $M^*$: this illustrates $|\OPT
(F_2)|\geq 2$. Matching $M_2$, being a $1/2$ approximation, will find at
least one $3$-augmenting path. \label{figure:two-pass-rand-bip} } 
\end{figure}

\begin{theorem}
\label{l_theorem2} Algorithm~\ref{algorithm:two-pass-rand-bip} is a randomized two-pass semi-streaming 
algorithm  for \mbmshort{} with expected approximation ratio $\frac{1}{2} + 0.019$
in expectation over its internal random coin flips for any graph and any arrival order,
and can be implemented with $\Order(1)$ update time. 
\end{theorem}
\begin{proof}
By construction, each edge in $M_2$ is part of a $3$-augmenting path, hence the output has size:
$ |M| = |M_0| + |M_2|. $
Define $\epsilon$ to be such that $|M_0| = (\frac{1}{2} + \epsilon) |\OPT(G)|$. 
Since $M_2$ is a maximal matching of $F_2$, we have $|M_2| \ge \frac{1}{2} |\OPT(F_2)|$.
Let $M^*$ be a maximum matching of $G$.
Then $|\OPT(F_2)|$ is greater than or equal to the number of 
edges $ab$ of $M_0$ such that there exists an edge $bc$ of $M_1$ 
and an edge $da$ of $M^*$ that altogether form a 3-augmenting path of $M_0$:
\begin{eqnarray*}
|\OPT(F_2)| & \ge 
   & |\{ ab \in M_0 \, | \, \exists c: bc \in M_1  \mbox { and } \exists d: da \in M^* \}|  \\ 
& \ge & |\{ ab \in M_0 \, | \, \exists c: bc \in M_1 \}| - |\{ab \in M_0 \, | \,  ab \mbox{ not 3-augmentable} \}| .
\end{eqnarray*}
Lemma~\ref{lemma:augmentable_edges} gives   
$|\{ab \in M_0 \, | \,  ab \mbox{ is not 3-augmentable with } M^*\}| \le 4 \epsilon |\OPT(G)|$.
It remains to bound $|\{ ab \in M_0 \, | \, \exists c: bc \in M_1 \}|$ from below. By definition of $M'$ and of $M_1\subseteq M'$, and by maximality of $M_0$, 
\begin{eqnarray*}
|\{ ab \in M_0 \, | \, \exists c: bc \in M_1 \}| & = & |M'| - |\{ ab\in M' \, | \, a\in A(M_0) \}|  \\
& \ge & |M'| - |A(M_0) \cap A'| .
\end{eqnarray*}
Taking expectations, by Theorem~\ref{theorem:subset-match} and by independence of $M_0$ from $A'$:
\begin{eqnarray*}
\Exp_{A'}|M'| - \Exp_{A'} |A(M_0) \cap A'| 
\ge  \frac{p}{1+p}|\OPT(G)| - p (\frac{1}{2} + \epsilon) |\OPT(G)| .
\end{eqnarray*}
Combining:
\begin{eqnarray*}
\Exp_{A'} |M| \ge (\frac{1}{2} + \epsilon) |\OPT(G)| + \frac{1}{2} \left( |\OPT(G)| p (\frac{1}{1+p} - \frac{1}{2} - \epsilon)
- 4\epsilon |\OPT(G)|. \right) 
\end{eqnarray*}
For $\epsilon$ small, the right hand side is maximized for $p = \sqrt{2} - 1$. Then $\epsilon \approx 0.019$ minimizes $\max\{ |M|, |M_0| \}$ which proves the theorem. \qed
\end{proof}

\section{Deterministic Two-pass Algorithm on any Order} \label{section:matching-two-pass-det}
We discuss now deterministic two-pass streaming algorithms for \mbm{} and \mm{} for input streams in adversarial order.
We start our presentation with an algorithm for bipartite graphs in Section~\ref{section:matching-two-pass-det-bip}. Then, we
show how this idea can be extended to general graphs in Section~\ref{section:matching-two-pass-det-gen}.

\subsection{Bipartite Graphs} \label{section:matching-two-pass-det-bip}
\subsubsection{Algorithm}
The deterministic two-pass algorithm, Algorithm~\ref{algo:two-pass-det-bip}, 
follows the same line as its randomized version, Algorithm~\ref{algorithm:two-pass-rand-bip}. 
In a first pass, we compute a Greedy matching $M_0$ and some additional edges $S$ that we compute by 
Algorithm~\ref{algo:incomplete-b-semi-matching}.
If $M_0$ is close to a $\frac{1}{2}$-approximation then $S$ contains edges that serve
as parts of $3$-augmenting paths. These are completed via a Greedy matching in the second pass.

The way we compute the edge set $S$ is now different. In Algorithm~\ref{algorithm:two-pass-rand-bip}, $S$ was a matching $M'$ between $B$ and 
a random subset $A'$ of $ A$.
Now, $S$ is not a matching but a relaxation of matchings as follows.
Given an integer $\lambda\geq 2$, an {\em incomplete $\lambda$-bounded semi-matching} $S$ of a bipartite graph $G=(A, B, E)$ 
 is a subset $S \subseteq E$ such that
 $ \deg_S(a) \le 1$ and $ \deg_S(b) \le \lambda$, for all $a\in A$ and $b\in B$.
This notion is closely related to semi-matchings. A semi-matching matches all $A$ vertices to $B$ 
vertices without limitations on the degree of a $B$ vertex. However, since we require that 
the $B$ vertices have constant degree, we loosen the condition that all $A$ vertices need to be matched.

In Lemma~\ref{lemma:incomplete-semi-matching}, we show that 
Algorithm~\ref{algo:incomplete-b-semi-matching}, a straightforward greedy algorithm, computes 
an incomplete  $\lambda$-bounded semi-matching that covers at least $\frac{\lambda}{\lambda+1}|M^*|$ vertices of $A$.
\begin{algorithm}
 \caption{Incomplete $\lambda$-bounded Semi-matching iSEMI($\lambda$) \label{algo:incomplete-b-semi-matching}}
  \begin{algorithmic}  
   \STATE $S \gets \varnothing$
   \STATE \textbf{while} $\exists $ edge $ab$ in stream
      \STATE \hspace{0.2cm} \textbf{if} $\deg_S(a) = 0$ and $\deg_S(b) \le \lambda-1$ \textbf{then} $S \gets S \cup \{ ab \}$
   \RETURN $S$
  \end{algorithmic}
\end{algorithm}
Now, assume that the greedy matching algorithm computes a $M_0$ close to a $\frac{1}{2}$-approximation.
Then, for $\lambda \ge 2$ there are many $A$ vertices that are not matched in $M_0$ but are matched in $S$. 
Edges incident to those in $S$  are candidates for the construction of $3$-augmenting paths. 
This argument can be made rigorous, leading to Algorithm~\ref{algo:two-pass-det-bip} where $\lambda$ is set to $3$,
see Theorem~\ref{theorem:two-pass-det-bip}.

\begin{algorithm} 
\caption{Two-pass Deterministic Bipartite Matching Algorithm \label{algo:two-pass-det-bip}}
 \begin{algorithmic}
 \STATE \textbf{First pass:} $M_0 \gets $ Greedy$(\pi)$ and $S \gets $ iSEMI$(3)$
 \STATE $S_1 \gets \{ e \in S \, | \, e = ab \mbox{ such that } a \in \overline{A(M_0)} \mbox{ and } b \in B(M_0) \}$
 \STATE $A_2 \gets \{ a \in A(M_0) \, | \, \exists bc: ab \in M_0 \text{ and } bc \in S_1 \} $
 \STATE $F \gets \{e \, | \, e=ab \mbox{ such that } a \in A_2 \mbox{ and } b \in \overline{B(M_0)} \}$
 \STATE \textbf{Second pass:} $M_2 \gets $ Greedy($\pi \cap F$)
 \STATE  Augment $M_0$ by edges in $S_1$ and $M_2$ and store it in $M$

 \RETURN $M$
\end{algorithmic}
\end{algorithm}
We show two figures illustrating the first pass (Figure~\ref{figure:two-pass-det-bip-1})
and the second pass (Figure~\ref{figure:two-pass-det-bip-2}) of Algorithm~\ref{algo:two-pass-det-bip}.
\begin{figure}[ht!]
\centerline{\scalebox{0.9}{
\begin{tikzpicture}[node distance = 1.5 cm]
    \GraphInit[vstyle=normal]
\draw(-2.0, 8.75) node{$3$-augmentable edges};
\draw(-2.0, 5.5) node{other edges};
\draw(2.25, 11) node{\textbf{$M_0$}};
\draw(4.35, 11) node{\textbf{$S_1 \subseteq S$}};
\draw(0, 10.5) node{\textbf{B}};
\draw(1.5, 10.5) node{\textbf{A}};
\draw(3, 10.5) node{\textbf{B}};
\draw(4.5, 10.5) node{\textbf{A}};
\draw(2.25, 7) node{$\vdots$};
\draw(2.25, 5.5) node{$\vdots$};
%
\draw(8.3, 8.5) node{edge $\in S_1$};
\draw(8.3, 8) node{edge $\in S \setminus S_1$};
\draw(8.3, 7.5) node{edge $\in M^*$};
\tikzset{VertexStyle/.style = {shape = circle, color = black, fill = white, minimum size = 0.1cm, draw}}
\Vertex[x=0, y=10, L=$$]{d1} \Vertex[x=1.5, y=10, L=$$]{a1} \Vertex[x=3, y=10, L=$$]{b1} 
\Vertex[x=0, y=9.5, L=$$]{d2} 				 \Vertex[x=3, y=9.5, L=$$]{b2} \Vertex[x=4.5, y=9.5, L=$$]{c2}
\Vertex[x=0, y=9, L=$$]{d3}  \Vertex[x=3, y=9, L=$$]{b3} \Vertex[x=4.5, y=9, L=$$]{c3}                
\Vertex[x=0, y=8.5, L=$$]{d4} \Vertex[x=1.5, y=8.5, L=$$]{a4} \Vertex[x=3, y=8.5, L=$$]{b4}                 
\Vertex[x=0, y=8, L=$$]{d5} \Vertex[x=1.5, y=8, L=$$]{a5} \Vertex[x=3, y=8, L=$$]{b5} \Vertex[x=4.5, y=8, L=$$]{c5}                
\Vertex[x=0, y=7.5, L=$$]{d6} 				 \Vertex[x=3, y=7.5, L=$$]{b6} \Vertex[x=4.5, y=7.5, L=$$]{c6}
\Vertex[x=0, y=6.5, L=$$]{d7} \Vertex[x=1.5, y=6.5, L=$$]{a7} \Vertex[x=3, y=6.5, L=$$]{b7} \Vertex[x=4.5, y=6.5, L=$$]{c7}
\Vertex[x=1.5, y=6, L=$$]{a8} \Vertex[x=3, y=6, L=$$]{b8} 
\Vertex[x=1.5, y=5, L=$$]{a9} \Vertex[x=3, y=5, L=$$]{b9} 
\Vertex[x=0, y=6, L=$$]{d8}
\tikzset{VertexStyle/.style = {shape = circle, color = black, fill = white, minimum size = 0.1cm, draw}}
\Vertex[x=1.5, y=9.5, L=$$]{a2}
\Vertex[x=1.5, y=7.5, L=$$]{a6}
\Vertex[x=4.5, y=10, L=$$]{c1}
\Vertex[x=4.5, y=8.5, L=$$]{c4}
%
\Vertex[x=1.5, y=9, L=$$]{a3}
\tikzset{VertexStyle/.style = {shape = circle, color = white, fill = white, minimum size = 0.1cm, draw}}
\Vertex[x=5.6, y=8.5, L=$$]{e2}
\Vertex[x=7.1, y=8.5, L=$$]{f2}
\Vertex[x=7.1, y=7.5, L=$$]{f3}
\Vertex[x=5.6, y=7.5, L=$$]{e3}
\Vertex[x=7.1, y=8, L=$$]{f4}
\Vertex[x=5.6, y=8, L=$$]{e4}
\SetUpEdge[style={solid}, color=black]
\Edge(a1)(b1)
\Edge(a2)(b2)
\Edge(a3)(b3)
\Edge(a4)(b4)
\Edge(a5)(b5)
\Edge(a6)(b6)
\Edge(a7)(b7)
\Edge(a8)(b8)
\Edge(a9)(b9)
\SetUpEdge[style={dashed}, color=black]
\Edge(d1)(a1) \Edge(b1)(c1)
\Edge(d2)(a2) \Edge(b2)(c2)
\Edge(d3)(a3) \Edge(b3)(c3)
\Edge(d4)(a4) \Edge(b4)(c4)
\Edge(d5)(a5) \Edge(b5)(c5)
\Edge(d6)(a6) \Edge(b6)(c6)
\Edge(d7)(a7) \Edge(b7)(c7)
\Edge(e3)(f3)
\Edge(a8)(d8)
\SetUpEdge[style={thick}, color=black]
\Edge(c2)(b1)
\Edge(c3)(b1)
\Edge(c1)(b3)
\Edge(c4)(b5)
\Edge(c6)(b5)
\Edge(e2)(f2)
\SetUpEdge[style={dotted}, color=black]
\Edge(a5)(b7)
\Edge(a8)(b7)
\Edge(a1)(d5)
\Edge(a2)(b3)
\Edge(a6)(b8)
\Edge(a3)(d5)
\Edge(e4)(f4)
 \end{tikzpicture}}}
\caption[First Pass of the Deterministic Two-pass Bipartite Matching Algorithm]{Illustration of the first pass of Algorithm~\ref{algo:two-pass-det-bip}. In this example we set $\lambda = 2$ and we compute an incomplete degree $2$ limited semi-matching $S$. By Lemma~\ref{lemma:incomplete-semi-matching}, we match at least $\frac{2}{3} |M^*|$ $A$ vertices. Since $|M| \approx \frac{1}{2} |M^*|$, some $A$ vertices that are not matched in $M_0$ are matched in $S$. The edges incident to those define $S_1$. \label{figure:two-pass-det-bip-1} } \end{figure}
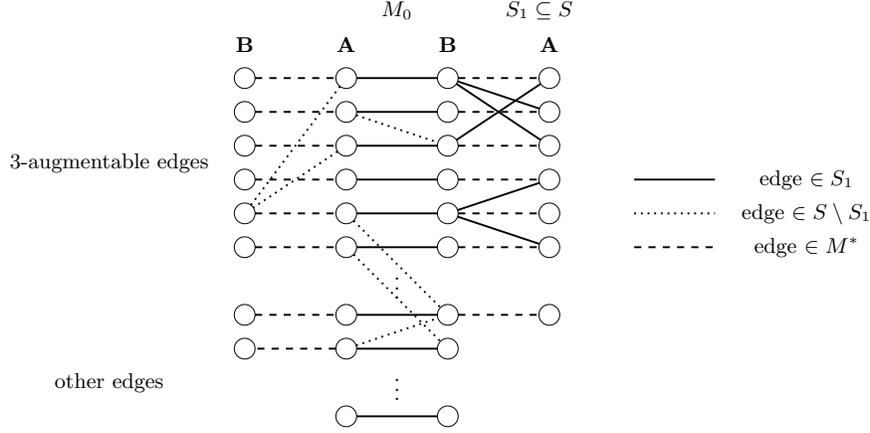

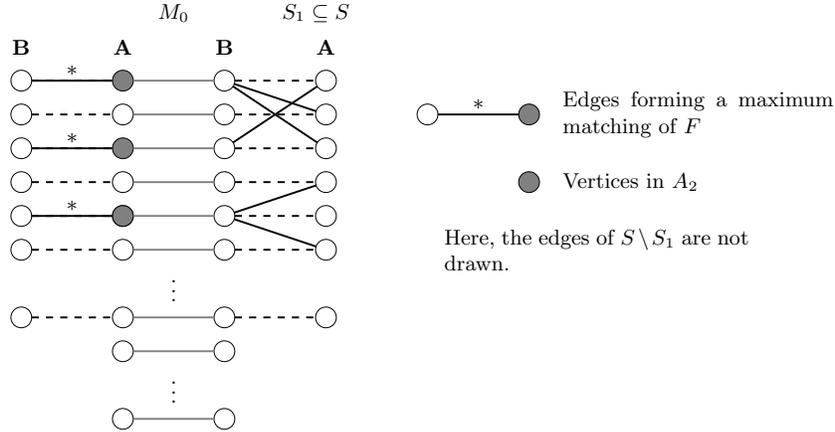
\begin{figure}[ht!]
\centerline{\scalebox{0.9}{
\begin{tikzpicture}[node distance = 1.5 cm]
    \GraphInit[vstyle=normal]
\draw(2.25, 11) node{\textbf{$M_0$}};
\draw(4.35, 11) node{\textbf{$S_1 \subseteq S$}};
\draw(0, 10.5) node{\textbf{B}};
\draw(1.5, 10.5) node{\textbf{A}};
\draw(3, 10.5) node{\textbf{B}};
\draw(4.5, 10.5) node{\textbf{A}};
\draw(2.25, 7) node{$\vdots$};
\draw(2.25, 5.5) node{$\vdots$};
\draw(6.75, 9.6) node{*};
\draw(0.75, 10.1) node{*};
\draw(0.75, 9.1) node{*};
\draw(0.75, 8.1) node{*};
\draw(8.5, 7.5) node{\parbox{4.5cm}{Here, the edges of $S \setminus S_1$ are not drawn.}};
\draw(10, 9.5) node{\parbox{4cm}{Edges forming a maximum matching of $F$}};
\draw(10, 8.5) node{\parbox{4cm}{Vertices in $A_2$}};
\tikzset{VertexStyle/.style = {shape = circle, color = black, fill = white, minimum size = 0.1cm, draw}}
\Vertex[x=0, y=10, L=$$]{d1} \Vertex[x=3, y=10, L=$$]{b1} 
\Vertex[x=0, y=9.5, L=$$]{d2} 				 \Vertex[x=3, y=9.5, L=$$]{b2} \Vertex[x=4.5, y=9.5, L=$$]{c2}
\Vertex[x=0, y=9, L=$$]{d3}  \Vertex[x=3, y=9, L=$$]{b3} \Vertex[x=4.5, y=9, L=$$]{c3}                
\Vertex[x=0, y=8.5, L=$$]{d4} \Vertex[x=1.5, y=8.5, L=$$]{a4} \Vertex[x=3, y=8.5, L=$$]{b4}                 
\Vertex[x=0, y=8, L=$$]{d5} \Vertex[x=1.5, y=8, L=$$]{a5} \Vertex[x=3, y=8, L=$$]{b5} \Vertex[x=4.5, y=8, L=$$]{c5}                
\Vertex[x=0, y=7.5, L=$$]{d6} 				 \Vertex[x=3, y=7.5, L=$$]{b6} \Vertex[x=4.5, y=7.5, L=$$]{c6}
\Vertex[x=0, y=6.5, L=$$]{d7} \Vertex[x=1.5, y=6.5, L=$$]{a7} \Vertex[x=3, y=6.5, L=$$]{b7} \Vertex[x=4.5, y=6.5, L=$$]{c7}
\Vertex[x=1.5, y=6, L=$$]{a8} \Vertex[x=3, y=6, L=$$]{b8} 
\Vertex[x=1.5, y=5, L=$$]{a9} \Vertex[x=3, y=5, L=$$]{b9} 
\Vertex[x=6, y=9.5, L=$$]{e2}
\tikzset{VertexStyle/.style = {shape = circle, color = black, fill = white, minimum size = 0.1cm, draw}}
\Vertex[x=1.5, y=9.5, L=$$]{a2}
\Vertex[x=1.5, y=7.5, L=$$]{a6}
\Vertex[x=4.5, y=10, L=$$]{c1}
\Vertex[x=4.5, y=8.5, L=$$]{c4}
\tikzset{VertexStyle/.style = {shape = circle, color = black, fill = gray, minimum size = 0.1cm, draw}}
\Vertex[x=1.5, y=10, L=$$]{a1} 
\Vertex[x=1.5, y=9, L=$$]{a3}
\Vertex[x=1.5, y=8, L=$$]{a5}
\Vertex[x=7.5, y=9.5, L=$$]{f2}
\Vertex[x=7.5, y=8.5, L=$$]{f3}
\tikzset{VertexStyle/.style = {shape = circle, color = white, fill = white, minimum size = 0.1cm, draw}}
\SetUpEdge[style={solid}, color=gray]
\Edge(a1)(b1)
\Edge(a2)(b2)
\Edge(a3)(b3)
\Edge(a4)(b4)
\Edge(a5)(b5)
\Edge(a6)(b6)
\Edge(a7)(b7)
\Edge(a8)(b8)
\Edge(a9)(b9)
\SetUpEdge[style={dashed}, color=black]
\Edge(d1)(a1) \Edge(b1)(c1)
\Edge(d2)(a2) \Edge(b2)(c2)
\Edge(d3)(a3) \Edge(b3)(c3)
\Edge(d4)(a4) \Edge(b4)(c4)
\Edge(d5)(a5) \Edge(b5)(c5)
\Edge(d6)(a6) \Edge(b6)(c6)
\Edge(d7)(a7) \Edge(b7)(c7)
\SetUpEdge[style={thick}, color=black]
\Edge(c1)(b3)
\Edge(c2)(b1)
\Edge(c3)(b1)
\Edge(c4)(b5)
\Edge(c6)(b5)
\Edge(a1)(d1)
\Edge(a3)(d3)
\Edge(a5)(d5)
\Edge(e2)(f2)
 \end{tikzpicture} }}
\caption[Second Pass of the Deterministic Two-pass Bipartite Matching Algorithm]{Analysis of the second pass of Algorithm~\ref{algo:two-pass-det-bip}. In this example, we set $\lambda = 2$. 
Here, we see that $M_0 \oplus
S_1$ has five paths of length 2. These paths are not disjoint, but since the maximal degree in $S$ is $2$, 
$M_0 \oplus S_1$ has at least $\frac{1}{2} \cdot 5$ disjoint paths, and hence $|A_2| = 3 \ge \frac{1}{2} \cdot 5$.
A maximum matching in $F$ is of size $3$, and in the second pass, Greedy will find at least half of them leading
to at least two $3$-augmenting paths.
 \label{figure:two-pass-det-bip-2} } 
\end{figure}

\subsubsection{Analysis}
We firstly present a lemma, Lemma~\ref{lemma:incomplete-semi-matching}, that analyses Algorithm~\ref{algo:incomplete-b-semi-matching}. This lemma
is then used in the proof of the main theorem, Theorem~\ref{theorem:two-pass-det-bip}.
\begin{lemma} \label{lemma:incomplete-semi-matching}
Let $S = \mathrm{iSEMI}(\lambda)$ be the output of Algorithm~\ref{algo:incomplete-b-semi-matching} for some $\lambda\geq 2$.
Then $S$ is an incomplete  $\lambda$-bounded semi-matching such that 
$|A(S)| \ge \frac{\lambda}{\lambda+1} |M^*|$.
\end{lemma}
\begin{proof}
By construction, $S$ is an incomplete degree $\lambda$ bounded semi-matching. We bound 
$A(M^*)~\setminus~A(S)$ from below. Let $a \in A(M^*) \setminus A(S)$ and let $b$ be 
its mate in $M^*$. The algorithm did not add the optimal edge $ab$ upon its arrival. 
This implies that $b$ was already matched to $\lambda$ other vertices. Hence,
$ |A(M^*) \setminus A(S)| \le \frac{1}{\lambda} |A(S)| $.
Then the result follows by combining this inequality with
$ |M^*| - |A(S)| \le |A(M^*) \setminus A(S)| $. \qed
\end{proof}

\begin{theorem}
 \label{theorem:two-pass-det-bip}
 Algorithm~\ref{algo:two-pass-det-bip} is a deterministic two-pass semi-streaming algorithm for 
\mbmshort{} with approximation ratio $\frac{1}{2} + 0.019$ for any graph and any arrival order
and can be implemented with $\Order(1)$ update time.
\end{theorem}

\begin{proof}
 The computed matching $M$ is of size $|M_0| + |M_2|$ since, by construction, for each
edge in $M_2$ there is at least one distinct edge in $S_1$ that allows the construction of
a $3$-augmenting path. Each $3$-augmenting path increases the matching $M_0$ by $1$. 
See also Figure~\ref{figure:two-pass-det-bip-2}.
 Since $|M_2|$ is a maximal matching of the graph induced by the edges $F$, we obtain
 \begin{eqnarray*}
 |M| \ge |M_0| + \frac{1}{2} |\OPT(F)| .
\end{eqnarray*}
Let $\epsilon$ be such that $|M_0| = (\frac{1}{2} + \epsilon) |M^*|$. By Lemma~\ref{lemma:augmentable_edges}, at most $4\epsilon|M^*|$ edges
of $M_0$ are not $3$-augmentable, hence
\begin{equation*}
 \OPT(F) \ge |A_2| - 4\epsilon |M^*|. 
\end{equation*}
$A_2$ are those vertices matched also by $M_0$ such that there exists an edge in $S_1$
matching the mate of the $A_2$ vertex. Since the maximal degree in $S_1$ is $\lambda$, 
we can bound $|A_2|$ by
\begin{eqnarray*}
 |A_2| \ge \frac{1}{\lambda}  |S_1|.
\end{eqnarray*}
Note that $|S_1| = |A(S) \setminus A(M_0)|$ since the degree of an $A$ vertex matched by $S$ in $S$ is one, and $S$ can be partitioned into $S_{M_0}, S_{\overline{M_0}}$ such that edges in $S_{M_0}$ couple an $A$ vertex also matched in $M_0$, and edges in $S_{\overline{M_0}}$ couple an $A$ vertex that is not matched in $M_0$. Now, $|S_1| = |S_{\overline{M_0}}|$ since an edge of $S$ is taken into $S_1$ if it is in $S_{\overline{M_0}}$. 
Lemma~\ref{lemma:incomplete-semi-matching}
allows us to bound the size of the set $A(S) \setminus A(M_0)$ via
\begin{eqnarray*}
 |A(S) \setminus A(M_0)| \ge |A(S)| - |A(M_0)| \ge (\frac{\lambda}{\lambda+1} - \frac{1}{2} - \epsilon)|M^*|.
\end{eqnarray*}
Using the prior Inequalities, we obtain
\begin{eqnarray*}
\label{equation:10}
 |M| \ge (\frac{1}{2} - \epsilon + \frac{1}{2\lambda+2} - \frac{1}{4\lambda} - \frac{\epsilon}{2\lambda}) |M^*| .
\end{eqnarray*}

Since we have also $|M| \ge |M_0| = (\frac{1}{2} + \epsilon)|M^*|$, we set 
\begin{eqnarray*}
 \epsilon_0 & = & \argmin_{\epsilon} \max \{(\frac{1}{2} - \epsilon + \frac{1}{2\lambda+2} - \frac{1}{4\lambda} - \frac{\epsilon}{2\lambda}) |M^*|, (\frac{1}{2} + \epsilon)|M^*| \} \\
 & = & \frac{\lambda-1}{8\lambda^2+10\lambda+2},
\end{eqnarray*}
which is maximized for $\lambda = 3$ leading to an approximation factor of $\frac{1}{2} + \frac{1}{52} \approx \frac{1}{2} + 0.019$. 

Concerning the update time, note that once an edge is added in the second pass,
a corresponding $3$-augmenting path can be determined in time $\Order(1)$.  \qed
\end{proof}

\subsection{Extension to General Graphs} \label{section:matching-two-pass-det-gen}
\subsubsection{Algorithm}
The deterministic two-pass algorithm for general graphs follows the same line as the
deterministic two-pass algorithm for bipartite graphs. In the first pass, Greedy matching 
$M$ together with some additional edges $F$ are computed. $F$ forms an \textit{incomplete $b$-bounded forest}.
\begin{definition}[incomplete $b$-bounded forest]
 Given an integer $b$, an incomplete  $b$-bounded forest $F$ is a 
 cycle free subset of the edges of a graph $G = (V, E)$ 
with maximal degree $b$.
\end{definition}
If $F \oplus M$ contains $3$-augmenting paths, we augment
$M$ by a maximal set of disjoint $3$-augmenting paths and store the result in $M'$. 
Those edges of $F$ that were not used in the previous augmentation and that form
length-$2$ paths with edges of $M'$ are stored in $M_R$.
In a second pass, length-$2$ paths of $M' \cup M_R$ are completed to
$3$-augmenting paths by computing a matching $M_L$. A maximal set of disjoint
$3$-augmenting paths of $M' \cup M_L \cup M_R$ is then used to augment $M'$.
\begin{algorithm}
 \caption{$b$-bounded Forest: FOREST$(b)$ \label{algo:b-bounded-forest}}
 \begin{algorithmic}[1]
  \REQUIRE $b$
  \STATE $S \gets \varnothing$
  \WHILE {stream not empty}
  \STATE $uv \gets $ next edge in stream
  \IF{$(\deg_S(u) = 0$ and $\deg_S(v) \le b-1)$ \textbf{ or } $(\deg_S(u) \le b-1$ and $\deg_S(v) = 0 )$} 
   \STATE $S \gets S \cup \{ uv \}$
  \ENDIF
  \ENDWHILE
  \RETURN $S$
 \end{algorithmic}
\end{algorithm}
\begin{algorithm}
\caption{Two-pass Deterministic Matching Algorithm for General Graphs \label{algo:two-pass-det-gen}}
 \begin{algorithmic}[1]
 \REQUIRE $b$
 \STATE $Aug \gets \varnothing$
 \STATE $M \gets $ Greedy$()$ and $F \gets $ FOREST$(b)$ \COMMENT{\textbf{first pass}}
 \STATE \label{line:max-3-aug-1} $M' \gets $ $M$ augmented by a maximal set of $3$-augmenting paths in $M \oplus F$
 \STATE $M_R \gets $ maximal subset of $F$ such that $\forall uv \in M_R: u \in V(M')$ and $\deg_{M_R}(v) = 1$
 \STATE $V' \gets \{ v \in V(M') \, : \, v' = M'(v) \mbox{ and } \exists v'u \in M_R: u \notin V(M') \} $
 \WHILE[\textbf{second pass}]{stream not empty} \label{line:second-pass}
  \STATE $vw \gets $ next edge in stream
  \IF{$v \in V'$ and $w \notin V(M')$ and $vw$ completes a $3$-augmenting path with edges $uv \in M', tu \in M_R$}
    \STATE \label{line:add-augmenting-path} $Aug \gets Aug \cup \{vw, tu \}, $ remove all edges from $M_R$ incident to $u$
  \ENDIF
 \ENDWHILE
 \STATE $M'' \gets M'$ augmented with $Aug$
 \RETURN $M''$
\end{algorithmic}
\end{algorithm}
Algorithm~\ref{algo:b-bounded-forest} is a greedy algorithm that constructs a forest $F$ such that the maximal
degree of a node in $F$ is $b$, for some $b \ge 1$. For a large enough $b$, all but a small fraction of the
vertices of the graph are covered by an edge in $F$.

The situation of the algorithm after the first pass is illustrated in Figure~\ref{figure:two-pass-det-gen}. 
Note that $M_R$ is an incomplete $b$-bounded semi-matching in the induced bipartite graph with vertex sets 
$V \setminus V(M')$ and $V(M')$. 
\begin{figure}[ht]
\begin{center}
 \includegraphics[bb = 110 330 470 510, scale=0.8]{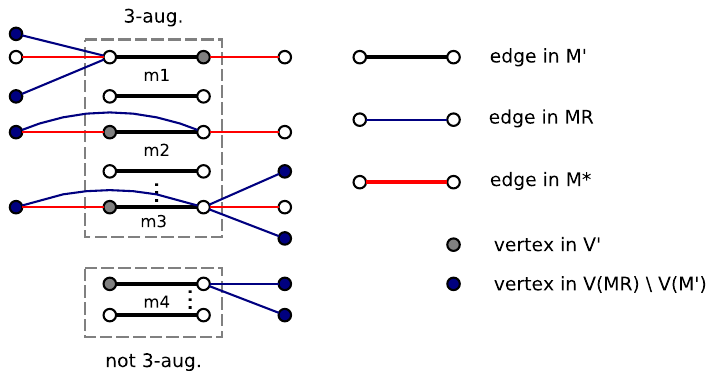}
\end{center}
\caption[Deterministic Two-pass Matching Algorithm for General Graphs]{Situation of Algorithm~\ref{algo:two-pass-det-gen} after
the first pass. $M'$ is the resulting matching that is obtained by augmenting $M$ with edges from $F$. $M_R$ is
a maximal subset of the edges of $F$ that were not used for augmenting $M$ such that vertices that are
free in $F$ with respect to $M'$ have a degree of one.
\label{figure:two-pass-det-gen} }
\end{figure}

\subsubsection{Analysis}
The analysis refers to the variables that are used in the algorithm. Furthermore, let $M^*$ denote a maximum matching
in the input graph and let $\epsilon$ be such that $|M| = (1/2 + \epsilon)|M^*|$. 
Let $\alpha = \frac{|M'|}{|M|} - 1$, or in other words, the set of disjoint $3$-augmenting paths 
found in Line~\ref{line:max-3-aug-1} is of size $\alpha |M|$. 

The analysis of the algorithm requires a lemma concerning the structure of forests. 

\begin{lemma}
 \label{lemma:forest}
 Let $T$ be a forest with at least $k$ nodes of degree at least $d$. Then:
\begin{equation*}
|T| \ge (d-1)k .  
\end{equation*}
\end{lemma}

\begin{proof}
 Consider the directed Graph $D$ that is obtained from $T$ by directing the edges from the roots of the trees of $T$ towards the leaves. 
Let $v_1, \dots, v_k$ denote the nodes that have degree at least $d$. 
Then for all $i \ne j: \Gamma_D(v_i) \cap \Gamma_D(v_j) = \varnothing$. Furthermore, for each $i: |\Gamma_D(v_i)| \ge (d-1)$. The result follows. \qed
\end{proof}


\begin{lemma}
\label{lemma:sizeS} Let $M^*$ denote a maximum matching in $G=(V, E)$. Consider the state of $F$ after the first pass. Then:
\begin{equation}
|F| \ge (b - 1) |V(M^*) \setminus V(F)|. 
\end{equation}
\end{lemma}

\begin{proof}
 By induction it is easy to see that $F$ is a forest with maximal degree $b$. 
We argue that $F$ has at least $|V(M^*) \setminus V(F)|$ nodes of degree $b$. 
The result then follows by applying Lemma~\ref{lemma:forest}.
Let $u \in V(M^*) \setminus V(F)$ and denote by $v$ the mate of $u$ in $M^*$. Since $uv$ is not taken, the degree of $v$ was already $b$ upon arrival of $uv$.
 Hence, for each node $u \in V(M^*) \setminus V(F)$ the partner $M^*(u)$ has degree $b$ in $F$. \qed
\end{proof}

\begin{lemma}
\label{lemma:freeS}
 Let $|M| = (\frac{1}{2} + \epsilon)|M^*|$. Consider the state of $F$ after the first pass. Then
\begin{equation*}
 |V(F) \setminus V(M)| \ge (1-2\epsilon - \frac{2}{b})|M^*|.
\end{equation*}
\end{lemma}

\begin{proof}
 By Lemma~\ref{lemma:sizeS}, $|F| \ge (b-1) |V(M^*) \setminus V(F)|$. Then
\begin{eqnarray}
\nonumber |V(F) \setminus V(M)| & \ge &  |V(F)| - |V(M)| \ge (b-1) |V(M^*) \setminus V(F)| - 2|M| \\
 & = & (b-1) |V(M^*) \setminus V(F)| - (1 + 2\epsilon) |M^*| . \label{eqn:3291}
\end{eqnarray}
Furthermore,we also have $|V(F)| \ge 2 |M^*| - |V(M^*) \setminus V(F)|$, and hence
\begin{eqnarray}
\nonumber |V(F) \setminus V(M)| & \ge & |V(F)| - |V(M)| \ge 2 |M^*| - |V(M^*) \setminus V(F)| - 2|M|  \\
\nonumber & = &  2 |M^*| - |V(M^*) \setminus V(F)| - (1+2\epsilon)|M^*| \\
& = &  (1-2\epsilon)|M^*| - |V(M^*) \setminus V(F)|. \label{eqn:9211}
\end{eqnarray}
Then $|V(M^*) \setminus V(F)| = \frac{2|M^*|}{b}$ minimizes $\max \{ \ref{eqn:3291}, \ref{eqn:9211} \}$ 
and we obtain $|V(F) \setminus V(M)| \ge \frac{2|M^*|}{b}$. \qed
\end{proof}

\begin{lemma} \label{lemma:large-continuation}
 Consider the state of the variables of the algorithm before the second pass. Let $M'_a \subseteq M'$
such that $\forall m \in M'_a$ there is an edge $m_R \in M_R$ and an edge $m_L \in E$ such that 
$m_R, m, m_L$ forms a $3$-augmenting path. Then:
\begin{equation*}
 |M'_a| \ge \frac{1}{b} \left( |V(M_R) \setminus V(M')| - |M'| \right) - 4(\epsilon + \frac{1}{2} \alpha + \alpha \epsilon)|M^*|.
\end{equation*}
\end{lemma}

\begin{proof}
The set $M'_a$ is precisely the subset of edges $uv$ of $M'$ that fulfill the following two conditions.
\begin{enumerate}
 \item $uv$ is $3$-augmentable, and 
 \item $uv$ has an edge of $M_R$ incident that is not a {\em blocking} edge. 
\end{enumerate}
We say that an edge $m_R = u'v \in M_R$ is a blocking edge, if $uv$ is the incident
edge of $M'$, $uu', vv'$ are the edges incident to $uv$ in $M' \oplus M^*$, and the edge
$u'v$ is not in the graph $G$. See Figure~\ref{figure:blocking-edge} for
an illustration. Note that there are at most $|M'|$ blocking edges in the graph.


We consider the vertices that are matched in $M_R$ but are free in $M'$. 
Each vertex $v \in V(M_R) \setminus V(M')$ is connected by an edge of $M_R$ to an edges of $M'$.
We remove from $V(M_R) \setminus V(M')$ these vertices that have a blocking edge incident.
There are at most $|M'|$ blocking edges. Since the maximal degree in $M_R$ is $b$, there are at least
$1/b ( |V(M_R) \setminus V(M')| - |M'|)$ edges in $M'$ that fulfill condition $(2)$. 
By Lemma~\ref{lemma:augmentable_edges}, there are at most 
$4(\epsilon + \frac{1}{2} \alpha + \alpha \epsilon)|M^*|$ edges in $M'$ that are
not $3$-augmentable, and the result follows. \qed

\end{proof}
\begin{figure}[ht]
\begin{center}
 \includegraphics[scale=0.8]{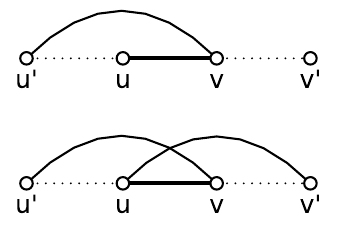}
\end{center}
\caption[Blocking Edge, Two-pass Deterministic Matching Algorithm for General Graphs]{Illustration of a blocking edge. In the first setting, the edge $u'v$ 
is a blocking edge, since the edge $uv'$ is not in the graph. The edge $u'v$ blocks edge $uu'$ from augmenting $uv$. 
In the second setting, neither $u'v$ nor $uv'$ are blocking edges. $u'v$ blocks the edge $vv'$, however, the edge $u'v$ is
an alternative for the node $v$ for being augmented. This alternative is not present in the first figure. \label{figure:blocking-edge} }
\end{figure}

\begin{theorem} \label{theorem:two-pass-det-gen}
 Algorithm~\ref{algo:two-pass-det-gen} with $b = 8$ is a deterministic $2$-pass semi-streaming algorithm for \mm~ 
with approximation ratio $1/2 + 1/140 \approx 1/2 + 0.007142$ for any graph and any arrival order.
\end{theorem}

\begin{proof}
By construction, the computed matching $M''$ is of size $|M'| + |Aug|$.  
Since $|M'| = (1+\alpha)|M|$ and $|M| = (\frac{1}{2} + \epsilon)|M^*|$, we obtain 
\begin{equation}
 |M''| = (1+\alpha) (\frac{1}{2} + \epsilon) |M^*| + |Aug|. \label{eqn:4935}
\end{equation}
 It remains to lower bound $|Aug|$.


In Lemma~\ref{lemma:large-continuation}, we show that there is a subset $M'_a \subseteq M'$ such that
\begin{equation*}
 |M'_a| \ge \frac{1}{b} \left( |V(M_R) \setminus V(M')| - |M'| \right) - 4(\epsilon + \frac{1}{2} \alpha + \alpha \epsilon)|M^*|,
\end{equation*}
and for each edge of $M'_a$ there is a $3$-augmenting path with an edge from $M_R$ and another edge from the stream. 
Any $3$-augmenting path that is added in Line~\ref{line:add-augmenting-path} of Algorithm~\ref{algo:two-pass-det-gen} to $Aug$ may block 
at most $2$ further edges of $M'_a$ from being augmented, see Figure~\ref{figure:half-augmentable}. 
We will find hence at least $\frac{1}{3} |M'_a|$ $3$-augmenting paths, and we obtain
\begin{equation} \label{eqn:3921}
 |Aug| \ge 1/3 |M'_a| \ge \frac{1}{3} \left( \frac{1}{b} \left( |V(M_R) \setminus V(M')| - |M'| \right) - 4(\epsilon + \frac{1}{2} \alpha + \alpha \epsilon)|M^*| \right) .
\end{equation}

Note that by construction, $|V(M_R) \setminus V(M')| = |V(F) \setminus V(M')|$.
We bound now $|V(F) \setminus V(M')|$. By Lemma~\ref{lemma:freeS}, $|V(F) \setminus V(M)| \ge (1 - 2\epsilon - \frac{2}{b})|M^*|$. 
Note that $M'$ is the matching that is obtained by augmenting $M$ with edges from $F$. Each augmented edge of $M$ has two
edges incident from $F$ that are used for the augmentation. Hence,
\begin{equation} \label{eqn:5222}
|V(F) \setminus V(M')| \ge (1 - 2\epsilon - \frac{2}{b})|M^*| - 2 \alpha |M| . 
\end{equation}

Using Inequality~\ref{eqn:5222} and Inequality~\ref{eqn:3921} in Inequality~\ref{eqn:4935}, we obtain

\begin{eqnarray}
 |M''| \ge \left(\frac{1}{2} + \frac{1}{6b} - \frac{1}{3}(\alpha \epsilon + \epsilon + \frac{\alpha}{2}) - \frac{1}{b} (\alpha \epsilon + \epsilon + \frac{\alpha}{2} + \frac{2}{3b})   \right)|M^*|. \label{eqn:4931}
\end{eqnarray}

Note that we also have
\begin{equation}
 |M''| \ge |M'| \ge |M^*| (\frac{1}{2} + \epsilon + \frac{\alpha}{2} + \alpha \epsilon). \label{eqn:4932}
\end{equation}

We determine $\epsilon_0$ as a function of $\alpha$ and $b$ that minimizes the maximum of the right sides of 
Inequality~\ref{eqn:4931} and Inequality~\ref{eqn:4932}. For any $\alpha$ and $\epsilon_0$, $M''$ is maximized by setting 
$b = 8$. This leads to an approximation factor $1/2 + 1/140 \approx 1/2 + 0.007142$. \qed
\end{proof}

\begin{figure}[ht]
\begin{center}
 \includegraphics[scale=0.8]{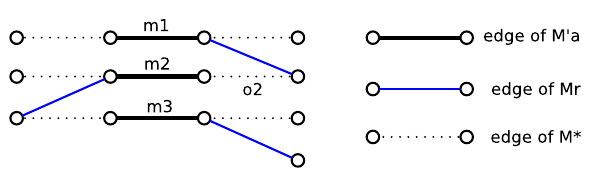}
\end{center}
\caption[An Augmenting Path blocks at most Two Augmentable Edges]{$m_1, m_2, m_3$ have each an edge of $M_R$ incident and can
be augmented with this edge and an incident edge from $M^*$. If $m_2$ is augmented with its incident edge from $M_R$ and
$o_2$, then this may prevent $m_1$ and $m_3$ from being augmented. \label{figure:half-augmentable} }
\end{figure}

\bibliography{kmm13.bib}

\end{document}